\documentclass[%
reprint,
amsmath,amssymb,aps,pra,
]{revtex4-2}

\usepackage{float,xcolor,soul,braket,hyperref,tikz,amssymb,amsthm,amsmath,mathtools,multirow,booktabs,bbm}
\usepackage[normalem]{ulem}
\usetikzlibrary{arrows.meta}
\usepackage[caption=false]{subfig}
\usepackage{ragged2e} 
\DeclareCaptionJustification{justified}{\justifying}
\sethlcolor{cyan}

\usepackage{graphicx}
\usepackage{dcolumn}
\usepackage{bm}
\newtheorem{theorem}{Theorem}

\newtheorem{lemma}[theorem]{Lemma}
\newtheorem{prop}[theorem]{Proposition}
\newtheorem*{prop*}{Proposition}
\theoremstyle{remark} 
\newtheorem{remark}[theorem]{Remark}

\renewcommand{\le}{\leqslant}

\renewcommand{\ge}{\geqslant}
\newcommand{\strat}{\mathfrak{S}}
\newcommand{\pA}{\tilde{A}}
\newcommand{\prbP}{\mathbb{P}}

\newcommand{\Rmnum}[1]{\MakeUppercase{\romannumeral#1}}
\newcommand{\SPs}[1]{\textcolor{blue}{\sout{#1}}}
\usepackage{xcolor}

\newcommand{\SPc}[1]{\textcolor{blue}{[SP: #1]}}

\newcommand{\PBc}[1]{\textcolor{teal}{[PB: #1]}}
\newcommand{\PB}[1]{\textcolor{teal}{#1}}

\usepackage{xcolor}

\begin{document}


\title{Minimal Scenarios Witnessing Genuine Multipartite Nonlocality Possibilistically}

\author{Soumyadip Patra}
\email{spatra@lsuneworleans.edu}
\author{Peter Bierhorst} 
\email{plbierho@lsuneworleans.edu}

\affiliation{Department of Mathematics, LSU New Orleans, Louisiana, USA}

\date{\today}

\begin{abstract}
We study a strong form of genuine multipartite nonlocality (GMNL), where an \(n\)-partite scenario admits a quantum strategy generating outcomes belonging to a restricted set with certainty, whereas every strategy built from nonclassical resources shared among at most \(k<n\) parties fails this criterion with positive probability. This generalizes what is referred to as strong contextuality/nonlocality, or pseudotelepathy, beyond ruling out local models. For the simplest five-party Bell scenario of two binary-outcome measurement settings per party, we use the inflation technique to demonstrate that a game with a perfect quantum strategy witnesses a robust separation from models built from bipartite resources under local operations and shared randomness (LOSR), establishing the behavior is LOSR-GMNL. We provide broader context for this result by investigating its simulability with resources shared among three or four parties under different resource models and alternative definitions of GMNL. A central open question is whether \(n\)-party quantum resources can achieve such a separation from models built from \((n-1)\)-party resources, for \(n\ge3\). We make progress here by constructing a family of robustly LOSR-GMNL games with \(\epsilon\)-optimal quantum strategies experimentally indistinguishable in the limit from true probability-one-winning strategies for any desired level of precision.
\end{abstract}

\maketitle

Nonlocality, a phenomenon in which an observed experimental behavior is probabilistically incompatible with any local hidden variable theory, has been demonstrated according to the most demanding experimental specifications~\cite{hensen:2015,abellan:2018,giustina:2015,shalm:2015}. Furthermore, by translating this incompatibility into quantitative limits on an adversary's knowledge, this foundational discovery can be applied to tasks like private randomness generation~\cite{pironio:2010,pironio:2013,bierhorst:2018}, secure key distribution~\cite{acin:2007}, and position verification~\cite{Fan-Yuan2026,Kavuri2026}.

A stronger form of nonlocality exhibits a \textit{possibilistic incompatibility}: only a certain subset of potential experimental outcomes are observed with positive probability, whereas any local hidden variable theory predicts at least one event assigned zero probability. This form of nonlocality is observable in many different configurations \cite{Greenberger1989,Mermin90,Aravind2002,Brassard99}, and the phenomenon has been referred to as \textit{strong contextuality}~\cite{Abramsky_2011,Mansfield17}, or 
\textit{pseudotelepathy}~\cite{Brassard2003,Brassard_2005} when quantum-achievable. This more powerful version of quantum nonlocality holds promise for applications such as zero-error communication through a noisy classical channel~\cite{Cubitt10} (see Proposition 3), certification of entangled states and measurements~\cite{Wu16} (see Appendix C), device-independent quantum key distribution~\cite{Zhen23}, and questions in quantum foundations~\cite{Barbosa26}.

For multipartite experiments, a natural question is whether an observed strongly nonlocal behavior rules out not only local hidden variable theories, but also more capable models allowing bipartite nonclassical resources but not higher-order multipartite nonclassical resources. Such behaviors are considered \textit{genuinely multipartite nonlocal}. A strongest-possible certificate of genuine multipartite nonlocality (GMNL) would be one following strong contextuality in which the observed behavior is restricted to a certain set of outcomes, while every model with at most bipartite nonclassical resources assigns a positive probability to some outcome outside the restricted set. This can be cast as a nonlocal game by considering outcomes in the restricted set to be ``success,'' such that the probability of success (or win probability) of a multipartite quantum strategy is one, whereas the restricted models have win probability less than one. Pursuing strongest certification of a network's underlying resource structure motivates asking whether quantum behaviors can enable such certification, and what operational advantage the certified structure might confer.

While the possibility of extending strong contextuality to GMNL has been explored before, such efforts have been somewhat disparate due to the different definitions of GMNL. The original Svetlichny definition~\cite{svetlichny:1987} classifies as GMNL  behaviors indecomposable as a probabilistic mixture of distributions that each factor into independent lower-partite distributions, originally even allowing signaling in these factors, but later refined to consider only no-signaling ones~\cite{Bancal13}. Such a decomposition was shown to exist in~\cite{Broadbent2006} for the perfect quantum-achievable correlations of the GHZ-Mermin game~\cite{Mermin90}, thus classifying it as non-GMNL. However, the absence of such a decomposition would not necessarily indicate GMNL under more stringent interpretations~\cite{Bierhorst21,schmid:2023,coiteux:2021}, as certain bipartite models, e.g., underlying ``wired" networks of bipartite Popescu-Rohrlich (PR) boxes~\cite{barrett:2005}, or parallel independent bipartite CHSH violations~\cite{schmid:2023}, can actually manifest Svetlichny-GMNL. Important previous work in~\cite{barrettpironio:2005} addresses this by exhibiting that a certain \(5\)-party nonlocal game, which has a quantum winning probability of one, is not perfectly winnable for a broad class of bipartite nonclassical resource models encompassing the examples above that are left out by the Svetlichny definition. The 
argument was later generalized in~\cite{Anshu13} to exhibit a family of pseudotelepathic behaviors that could rule out arbitrarily-large-partite ``wiring'' models---specifically, for any value of $k$ there is a $k^22^{2k-2}$-partite nonlocal game with a quantum winning probability of one, that is incompatible with any underlying network of $k$-partite no-signaling resources~\cite{Anshu13}. These results, however, were not noise-robust and do not rule out entangled measurements of bipartite quantum resources; such models can unlock new behaviors inaccessible to underlying wired networks of PR boxes~\cite{Bierhorst26} and can be considered non-GMNL \cite{coiteux:2021}. In contrast, a recent work~\cite{Brassard26} leverages the inflation technique~\cite{wolfe:2019} to prove that a different five-party pseudotelepathic game is local-operations-shared randomness (LOSR) GMNL. Introduced in~\cite{coiteux:2021}, LOSR-GMNL classifies as GMNL any behavior that cannot be replicated with any underlying network of at-most-bipartite nonclassical resources (possibly superquantum), allowing arbitrary measurements including entangled measurements or even non-quantum generalizations of the concept, and also permitting globally shared classical randomness; this encompasses all lower-partite models discussed above, and can be considered the most comprehensive definition of GMNL. In the scenario of~\cite{Brassard26}, each of the five parties has two settings, while two of them have eight outcomes and the others have two outcomes each. A nine-qubit quantum strategy judiciously measuring three Greenberger-Horne-Zeilinger (GHZ) states can win the game with probability one, while lower-partite models cannot, though
Ref.~\cite{Brassard26} did not derive an explicit robustness bound; robustness is necessary to distinguish the behavior from non-GMNL models experimentally. This work also leaves open the question as to whether simpler scenarios can witness the separation, a question of both practical and conceptual significance.



We demonstrate here that the simplest nontrivial 5-party scenario---the (5,2,2) scenario of five parties, two settings per party, and two outcomes per setting---admits what we will henceforth refer to as a \textit{strongly LOSR-GMNL} behavior: a pseudotelepathic game for which the corresponding behavior is LOSR-GMNL. We do this with a new inflation-technique argument applied to the 5-party game in~\cite{barrett:2006}, which we will refer to as the 5-cycle game following the convention of~\cite{Mathieu2018,Chao17}. We also derive explicit robustness bounds---a winning probability bound of \(121/126\approx0.960\) which improves on previously reported results for the 5-cycle game both in degree (lowering a bound of \(53/54\approx 0.981\) on wirings of bipartite nonclassical resources that was derived in \cite{Mathieu2018}) and scope (ruling out the larger, full class of theories demanded by the LOSR-GMNL definition). The explicit proof techniques are notable since the difficulty of a linear program that might \textit{in principle} generate robustness bounds from inflation technique arguments grows exponentially for multi-partite scenarios, quickly becoming intractable with standard computing resources. We also provide broader context for our results by performing a comprehensive study of the levels of GMNL for this game, finding that four-party nonclassical resources \textit{can} simulate the behavior, so that it is not genuinely \textit{five}-party nonlocal, while also finding it is not simulable with three-party resources at least according to a no-signaling-Svetlichny model. Aspects of our ladder of GMNL results for the 5-cycle game demonstrate intriguing parallels to Mermin's strongly contextual/pseudotelepathic GHZ game~\cite{Mermin90,Mermin93}, and our detailed study of a quantum strategy for optimizing the winning probability of the reduced 4-party version of the 5-cycle game demonstrates interesting properties that may be of independent interest.

The natural follow-up question, recognized in~\cite{Brassard26}, remains whether there exist $n$-party pseudotelepathic games for which no $(n-1)$-party resource model exists, for $n\ge3$. 
We close our paper by demonstrating a family of 3-party nonlocal games whose members admit quantum strategies with winning probabilities arbitrarily close to one, while proving that each exhibits LOSR-GMNL robustly. The argument is straightforwardly generalizable to an $\epsilon$-close-to-strong separation of $n$-party games from $(n-1)$-party strategies. While this does not completely resolve the question, these behaviors are in principle experimentally indistinguishable from true strong LOSR-GMNL in the sense that $\epsilon$ can be chosen small enough so that a losing trial is effectively never observed.

\medskip

\noindent\textit{Genuine multipartite nonlocality of the 5-cycle game}. In the 5-cycle game of Ref.~\cite{DiVincenzoPeres97}, there are five parties \(\mathsf A\), \(\mathsf B\), \(\mathsf C\), \(\mathsf D\), and \(\mathsf E\); each makes one of two measurements labelled \(X\) or \(Z\) and receives an outcome in the set \(\{0,1\}\). With \(R_Q\) denoting outcome \(R\) given setting \(Q\) for party \(\mathsf R\), the success conditions for six measurement configurations are given as the following predicates, where $\oplus$ denotes addition modulo 2:
\begin{equation}\label{e:5game}
\begin{alignedat}{2}
P_1\colon & A_X \oplus B_Z \oplus C_X = 1\,\,\,\,\, & P_2\colon & B_X \oplus C_Z \oplus D_X = 1\\
P_3\colon & C_X \oplus D_Z \oplus E_X = 1\,\,\,\,\, & P_4\colon & D_X \oplus E_Z \oplus A_X = 1\\
P_5\colon & E_X \oplus A_Z \oplus B_X = 1\\
P_6\colon & \rlap{$A_Z \oplus B_Z \oplus C_Z \oplus D_Z \oplus E_Z = 0.$}
\end{alignedat}
\end{equation}
 The game is highly symmetric among the parties, as predicates \(P_2\)--\(P_5\) are all cyclic permutations of \(P_1\). There is a five-qubit quantum strategy that satisfies all predicates with certainty~\cite{DiVincenzoPeres97} which we provide in the Supplemental Material (SM)~\ref{s:qstrat}. Conversely, the game is nonlocal~\cite{DiVincenzoPeres97} and more generally non-simulable with wirings of bipartite arbitrary no-signaling resources~\cite{barrettpironio:2005,Mathieu2018}, in that the winning probability for such models---the average probability of satisfying each of the six predicates when the corresponding measurement configurations are measured with equal probability---is less than one. 
A useful preliminary result, which illustrates the nonlocality of the game, is as follows:
\begin{prop}\label{p:2DetPlayers}
The maximum winning probability of the game when two players respond deterministically on input \(X\) is \(5/6\).
\end{prop}
\noindent Note that the bound on winning probability is \textit{not} claimed if we fix only a single player's output; we will see later that the proposition in fact does not hold under this weaker condition. A full proof of Proposition~\ref{p:2DetPlayers} is in SM~\ref{s:prop1proof}, but the basic idea is illustrated by the following argument that rules out simultaneous satisfaction of all predicates when adjacent parties \(\mathsf{A, B}\) deterministically output \(A_X=\alpha, B_X=\beta\). In this case, satisfaction of predicates \(P_1\) and \(P_5\) requires that \(B_Z\oplus C_X=\alpha\oplus 1\) and \(E_X\oplus A_Z=\beta\oplus 1\) hold with certainty, and since \(A_Z,B_Z,C_X,D_Z,E_X\) can be jointly measured, these facts combine with predicate \(P_3\) to imply \(B_Z\oplus D_Z\oplus A_Z=1\oplus\alpha\oplus\beta\) always holds. Moreover, predicates \(P_2\) and \(P_4\) respectively imply \(C_Z\oplus D_X = 1 \oplus \beta\) and \(D_X\oplus E_Z=1\oplus \alpha\), whereby the joint measureability of \(C_Z,D_X,E_Z\) implies \(C_Z\oplus E_Z=\alpha\oplus \beta\). The two results obtained so far combine to imply \(A_Z\oplus B_Z\oplus C_Z\oplus D_Z\oplus E_Z=1\), but this is incompatible with \(P_6\).

\begin{figure}[!htb]
\centering
\includegraphics[width=1.0\linewidth]{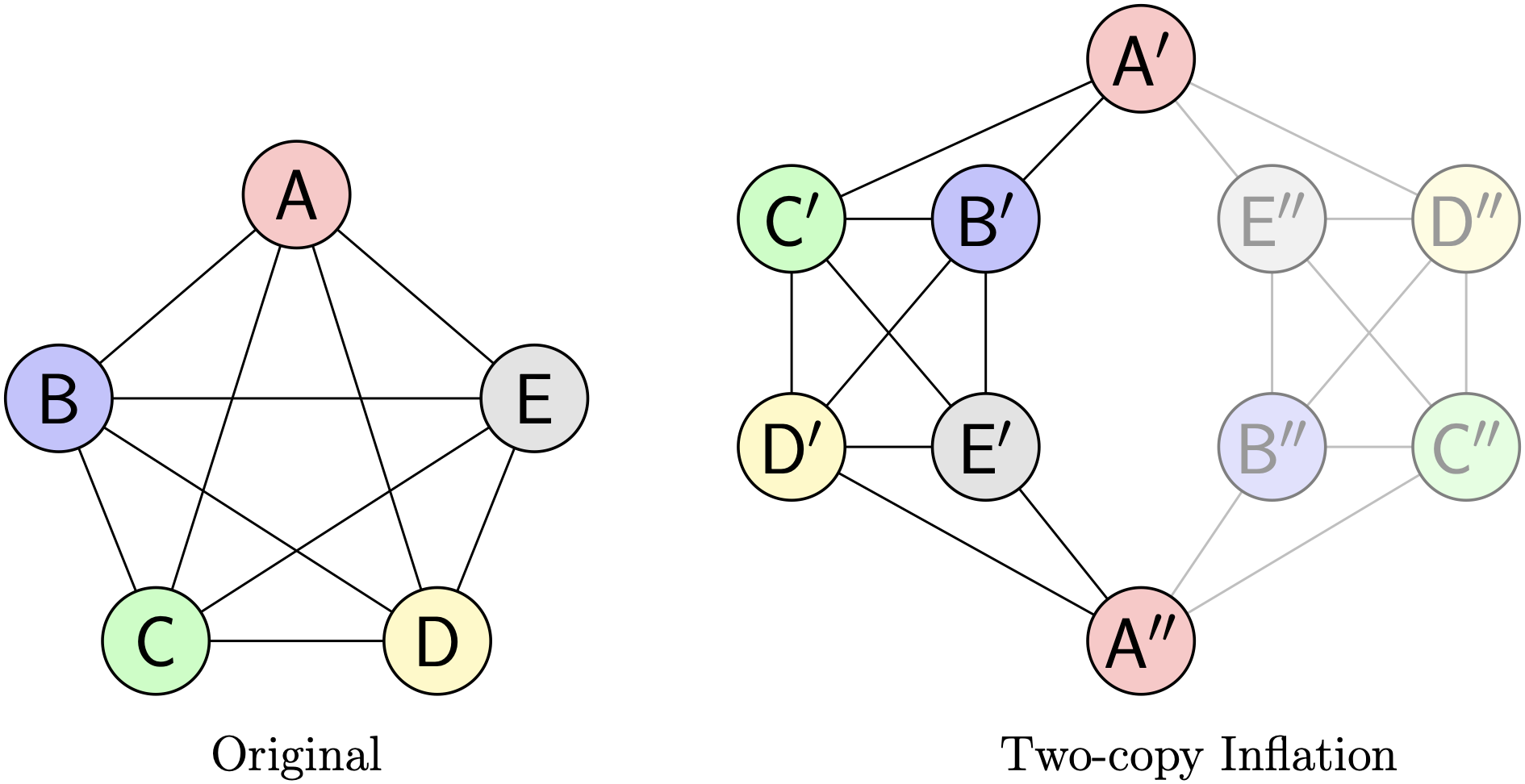}
\caption{The pairwise-resource network for the \(5\)-cycle game, and its two-copy non-fanout inflation used in our proof of Proposition~\ref{prop:ImprovedBndBipartiteRes}. Nodes and edges represent players and bipartite no-signaling resources, respectively. All players have access to globally shared classical randomness (not depicted). In the inflation network on the right, the left square with nodes \(\mathsf{B}',\mathsf{C}',\mathsf{D}',\mathsf{E}'\) gives a \((\mathsf{B}',\mathsf{C}',\mathsf{D}',\mathsf{E}')\)-marginal which is importable from the original network in the sense that its underlying resource network is the same and so generates the same marginal distribution. The triples \((\mathsf{A}',\mathsf{B}',\mathsf{C}')\) and \((\mathsf{D}',\mathsf{E}',\mathsf{A}'')\) are also importable, while \(\mathsf{A}'\) and \(\mathsf{A}''\) have disjoint source ancestry. The faded part of the diagram contains copies of players (nodes) and resources (edges) not used in our proof.}
\label{fig:Two-copyInflation}
\end{figure}
We can use Proposition~\ref{p:2DetPlayers} to show that any correlation satisfying all predicates in~\eqref{e:5game} with probability one is not just nonlocal but exhibits LOSR-GMNL, strengthening the results of~\cite{barrettpironio:2005,Mathieu2018} to rule out a broader class of bipartite-only-resource models. 
We use the inflation technique~\cite{wolfe:2019} to prove our result as follows. Assume, by way of contradiction, that all predicates in~\eqref{e:5game} hold with certainty for a model compatible with the underlying resource network on the left side of Figure~\ref{fig:Two-copyInflation}. Since this stipulation holds for the average over any globally shared classical randomness, it must hold for each individual value of the classical randomness, so we effectively ignore this aspect by assuming that the classical randomness takes one of the fixed values, and that this value is instantiated for both the original and inflated network. Now observe that predicates \(P_1\) and \(P_4\) can be simultaneously queried since \(A_X,B_Z,C_X,D_X,E_Z\) are jointly measurable; their simultaneous satisfaction then implies a key observation that \(B_Z\oplus C_X \oplus D_X \oplus E_Z=0\), since, depending on the value of \(A_X\), either \(B_Z\oplus C_X\) and \(D_X \oplus E_Z\) are both \(0\) or both \(1\). Now, moving from the left side of Figure~\ref{fig:Two-copyInflation} to the inflated copy on the right, we have $A'_X\oplus B'_Z \oplus C'_X=1$, $B'_Z\oplus C'_X \oplus D'_X \oplus E'_Z=0$, and $D'_X \oplus E'_Z\oplus A''_X=1 $ with certainty, because the underlying subnetworks are isomorphic to the original graph for these subsets of parties. As all three of these primed predicates are simultaneously observable in the two-copy inflation, they result in the implications \(A'_X=\alpha \Rightarrow B'_Z\oplus C'_X =1\oplus \alpha \Rightarrow D'_X\oplus E'_Z=1\oplus\alpha \Rightarrow A''_X=\alpha\), so that \(A'_X=A''_X\) with probability one. However, since \(A'_X\) and \(A''_X\) are causally disconnected, they must be independent; their equality then implies that they must be constant, and so by one-party-subnetwork isomorphism \(A_X\) must be constant as well. By repeating the above argument with an analogous two-copy inflation centered at \(B_X\), we find \(B_X\) is also constant, and thus, by Proposition~\ref{p:2DetPlayers}, at most five of the six predicates in~\eqref{e:5game} can be satisfied.

The above ``obstruction'' argument shows that multipartite nonclassical resources are necessary to simultaneously satisfy all predicates in~\eqref{e:5game}, but does not provide an explicit bound on the winning probability for the pairwise-resource model. For this we have the following: 


\begin{prop}\label{prop:ImprovedBndBipartiteRes}
The average winning probability of the 5-cycle game for every behavior in the pairwise-resource network model is bounded above by $121/126\approx0.960$.
\end{prop}
\noindent The proof of Proposition~\ref{prop:ImprovedBndBipartiteRes} is in SM~\ref{s:prop2proof}. It uses a higher-order inflation, and so the proof's methods for elevating an obstruction argument to a robustness bound may be useful for future studies of strong LOSR-GMNL, due to the inherently larger numbers of parties causing higher order inflations to grow exponentially in complexity, thereby constraining automated linear programming approaches. 

\medskip

\noindent\textit{Comprehensive analysis of genuine-multipartite nonlocality of 5-cycle game.} 
As mentioned earlier, if only a single player's output is fixed in Proposition~\ref{p:2DetPlayers}, a no-signaling nonlocal distribution over the other four parties \textit{does} allow for the predicates in~\eqref{e:5game} to be satisfied with probability \(1\). This  demonstrates the 5-cycle game is not genuinely five-partite nonlocal (i.e., does not require at least five-partite nonclassical resources). To elaborate, if we fix one player's outputs as \(A_X=r,A_Z=s\oplus 1\) for some choice of \(r,s\in\{0,1\}\), the predicates in~\eqref{e:5game} reduce in the following manner (for the example of \(r=s\oplus1=0\)):
\begin{equation}\label{e:4game}
\begin{split}
&P_1\colon B_Z\oplus C_X=1\,\,\,\,\,\,\,\,\,\,\,\,\,\,\,\,\,\,\,\,\,\, P_2\colon B_X\oplus C_Z\oplus D_X=1 \\
&P_3\colon C_X\oplus D_Z\oplus E_X=1\,\,\,\,\, P_4\colon D_Z\oplus E_Z=1 \\
&P_5\colon E_X\oplus B_X=1\,\,\,\,\,\,\,
P_6\colon B_Z\oplus C_Z\oplus D_Z\oplus E_Z=0.
\end{split}
\end{equation}
Exhibiting a 4-party no-signaling distribution that can simultaneously satisfy the above conditions is a tractable linear program which we solved; the results are in SM
Section~\ref{s:4-1split}. There we also describe how the four solutions to the linear program for each of the four choices of \(r,s\) can be adapted into four five-party distributions, each with constant $A_X=r$ and $A_Z=s\oplus 1$, such that when these are averaged together the mixture distribution does not just win the 5-cycle game but furthermore replicates the full statistics of the winning quantum distribution given in SM Section \ref{s:qstrat}. This is a no-signaling Svetlichny-type model, and so it refutes full five-partite GMNL according to both the Svetlichny as well as the more expansive LOSR-GMNL definitions. 

We follow up on the above finding by asking whether \textit{quantum} strategies can ``win" the four-player game of~\eqref{e:4game}, which is nonlocal with a local bound of \(5/6\)---if this were so, it would allow a simpler four-player game to witness strong LOSR-GMNL with a quantum realization. 
Intriguingly, we are able to bound the maximal quantum winning probability strictly below 1 (via the NPA hierarchy~\cite{navascues:2007,navascues:2008}), while exhibiting that quantum mechanics can beat the local bound with a value of \((4+\sqrt{2})/6\) (via an explicit example). Our NPA bound matches the explicit quantum value to eight decimal places, indicating this is the true quantum optimum. We describe these findings in fuller detail in SM Section~\ref{s:qtest}, where we illustrate compelling features of the quantum winning strategy that may be
of independent interest. 

The above results exhibit an intriguing parallel to the GHZ game of Mermin \cite{Mermin90}, which we recall is a nonlocal game for the simplest nontrivial three-party scenario of \(3\) parties, \(2\) settings, and \(2\) outcomes, or \((3,2,2)\), and can be won with probability \(1\) with a three-way quantum strategy measuring the GHZ state. If superquantum two-party nonclassical resources are allowed, only \textit{bipartite} no-signaling resources (PR boxes) are needed to exhibit this behavior~\cite{Broadbent2006}---but quantum requires the ``extra" third player to exhibit pseudotelepathy, since the \((2,2,2)\) PR box distribution is nonquantum. In parallel, the 5-cycle game involves the simplest nontrivial \textit{five}-party scenario \((5,2,2)\), the game can be won with probability \(1\) with a five-party quantum strategy, but only four-party no-signaling resources (albeit superquantum) are needed---once again, quantum requires the ``extra player," but now to demonstrate strong LOSR-GMNL as opposed to (only) the strong nonlocality exhibited by the GHZ-Mermin game. 

Rounding out the above considerations is the observation that \textit{three}-party resources are insufficient to win the 5-cycle game, at least for a no-signaling Svetlichny-type model, and so quantum mechanics does only require just the single extra player beyond the minimal-partite \(4\)-party no-signaling resource capable of winning the 5-cycle game. To see this, recall a three-party no-signaling Svetlichny decomposition of the behavior is a convex sum of factoring distributions, all of the form $P(AB)Q(CDE)$ or another analogous $2|3$ split among the parties, where each $P$ and $Q$ is a no-signaling distribution. Since a two-party no-signaling distribution $P$ can be expressed as convex combination of behaviors which are either local deterministic or PR boxes \cite{barrett:2005}, Proposition-\ref{p:2DetPlayers}-style reasoning shows that each factoring $P(AB)Q(CDE)$-type distribution obeys the winning probability bound of 5/6, and, therefore, so does the sum; the full argument is given in SM Section~\ref{s:svet}. We summarize our findings on the multipartite nonlocality of the 5-cycle game in Table~\ref{tab:Overview}, using ``$\Rightarrow$" to indicate GMNL results that imply other ones.
\begin{table}[!htb]
\centering
\caption{\label{tab:Overview}Overview of simulability of the 5-cycle game for different resource classes.}
\begin{tabular}{cccc}
\toprule
Resources & LOSR-GMNL & & Svetlichny-GMNL \\
\midrule
Bipartite & Not Simulable & \(\Longrightarrow\) & Not Simulable \\
 & & & \(\Uparrow\) \\
Tripartite & Open & & Not Simulable \\
 & & & \\
Quadpartite & Simulable & \(\Longleftarrow\) & Simulable \\
\bottomrule
\end{tabular}
\end{table}

\noindent\textit{A family of tripartite LOSR-GMNL games with \(\epsilon\)-close-to-winning quantum strategies, and generalizations to $n|(n-1)$-party separations.} 
While we do not know if there is a quantum-achievable 3-party strongly LOSR-GMNL behavior, we are able to construct $\epsilon$-close instances. 
Conceptual motivation for the significance of this partial result follows from the observation that while quantum pseudotelepathy is ostensibly the ability to exhibit a winning probability of one in a nonlocal game, there are no perfect experiments: 
for example, the recent pseudotelepathy experiment of Ref.~\cite{GoogleQuantumAI2025} observes a winning probability of \(0.983\) for the Mermin-Peres game~\cite{Mermin1990,Peres1990,Cabello2001_AvN} which exceeds the local bound of $8/9 \approx 0.889$---but in some sense, this observed separation is not qualitatively different from (say) a standard CHSH experiment observing a winning probability of 85\% in excess of the local bound of 75\% (where quantum mechanics cannot exceed Tsirelson's bound of $\sim 85.4\%$). 
Though loss probability can never be eliminated completely, future pseudotelepathy experiments---for which 
the defining characteristic is winning with certainty---can still aim to exhibit a winning probability very 
close to one, perhaps aspirationally 
at par with phenomena taken more or less to ``always" occur, like the probability that a pocket calculator correctly solves the query ``2+2'', for which there is \textit{some} nonzero probability of malfunction. 
Following this aspiration, if a family of games is only $\epsilon$-close to pseudotelepathic, $\epsilon$ can be chosen to make a game indistinguishable from pseudotelepathy as an experimentally observable phenomenon. 
Of course, limits in precision may make this choice of $\epsilon$ ultimately infeasible for achieving very-near-unit probability of winning, but ``true" pseudotelepathy scenarios, which may involve experimentally challenging complications of their own (such as two additional measuring parties for the example of this paper), are also unavoidably subject to precision considerations inevitably pushing the probability of winning away from 100\%. 

With this motivation in mind, we present an $N$-indexed family of LOSR-GMNL 3-party games with bipartite-resource-network winning probability bounded above by $1-\gamma/N$, and a tripartite quantum strategy with winning probability exceeding $1-\delta/N^2$, for explicit constants $\gamma=1/6$ and $\delta=\pi^2/32$. As we show in SM~\ref{s:limitingbehavior}, this implies that for any desired completeness parameter $\epsilon_c>0$ and soundness parameter $\epsilon_s>0$, there is a choice of game index $N$ and number of trials $t$ for which $1-\epsilon_c$ is a lower bound on probability that the honest implementation yields pseudotelepathic behavior of exclusively ``win" outcomes over $t$ trials, and $\epsilon_s$ is an upper bound on $t$ wins in $t$ trials for the foil theory (bipartite resource models). This in principle allows experiments observationally indistinguishable from true quantum telepathy, constrained only by the ultimate precision limits of the standard quantum mechanical model for Bell experiments.


\begin{figure}[!htb]
\centering
\includegraphics[width=1.0\linewidth]{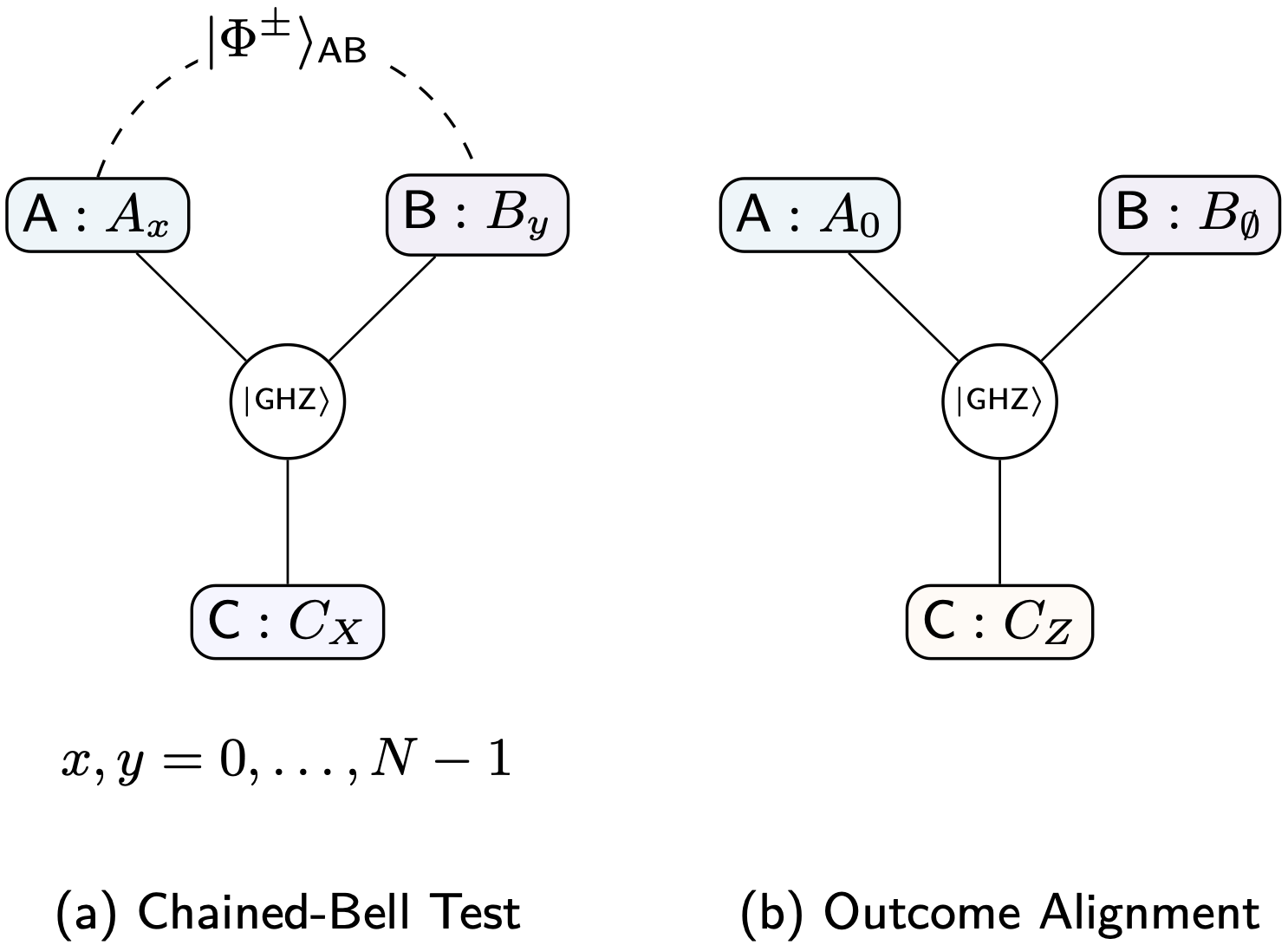}
\caption{Schematic depiction of the $N$-indexed family of chained Bell-GHZ games. The three parties $\mathsf {A,B,C}$ share a GHZ state; party $\mathsf C$ either (a) measures in basis $X$, which probabilistically swaps one of two Bell states $|\Phi^\pm\rangle=(|00\rangle\pm |11\rangle)/\sqrt{2}$ to parties $\mathsf A$ and $\mathsf B$, or (b) measures in basis $Z$, which allows perfect outcome alignment when the other parties also measure in $Z$. Depending on $N$, $\mathsf A$ and $\mathsf B$ perform measurements $\{0,...,N-1\}$ simultaneously encompassing two variants of the $N$-fold chained-Bell inequality that are violated in the $C_X$-scenario of part (a); we configure these chained-Bell measurements to naturally include $Z$ for $\mathsf A$ as $A_0=A_Z$ while appending an extra measurement $B_\emptyset=B_Z$ for $\mathsf B$, allowing perfect outcome agreement in the $C_Z$ scenario of part (b).}
\label{fig:ChainBellGHZ}
\end{figure}

The setup for our family of games is depicted in Figure~\ref{fig:ChainBellGHZ}. The intuition, reminiscent of the CHSH approaches of~\cite{coiteux:2021,mao:2022,cao:2022,Chao17} that bear some similarity to our game's $N=2$ case, is that for an underlying bipartite network model, $\mathsf A$ and $\mathsf B$ cannot effectively win a chained Bell game on $C_X$ while simultaneously aligning outcomes on $C_Z$. 
Concretely, the game is to simultaneously satisfy the following $4N+2$ predicates:
\begin{align}\label{e:cbgame}
\{C_X=0, A_i=B_{i-j}\}&\text{ OR }\{C_X=1, \text{ any }A_i,B_{i-j}  \};\nonumber\\
\{C_X=1, A_i\ne B_{N-i-j}\}&\text{ OR }\{C_X=0, \text{ any }A_i,B_{N-i-j}\};\nonumber\\
A_0=B_0;\;\; A_0\ne B_{N-1};&\;\;A_0=B_\emptyset; \;\; \;\;B_\emptyset=C_Z
\end{align}
for $i\in{1,...,N-1}$ and $j\in\{0,1\}$. We prove the claim that this is $\epsilon$-strongly LOSR-GMNL in SM Sections~\ref{s:cbinflationproof} and~\ref{ss:QuantumVal_ChainedBell}, demonstrating how \eqref{e:cbgame} accommodates two chained Bell variants tailored to $\ket{\Phi^+}$ and $\ket{\Phi^-}$ by defaulting to a universal $\mathsf {A} /\mathsf {B}$-win condition when the ``wrong" chained-Bell game is queried.
 We also describe in SM \ref{s:nsep} how to generalize to an $n|(n-1)$ LOSR-GMNL separation by measuring higher-order GHZ states and distributing the role of Charlie to multiple distinct Charlie-like players. 

\medskip

\noindent \textit{Conclusion}. We have robustly demonstrated a strongly LOSR-GMNL behavior in the simplest nontrivial five-party Bell scenario and provided context for this result with a comprehensive analysis of its full spectrum of multipartite nonlocality which reveals interesting parallels to the party-irreducibility of the (3,2,2) Mermin-GHZ game---perhaps suggesting an impediment to further party-number reductions for a demonstration of strongly LOSR-GMNL quantum behaviors in a minimal 2-setting, 2-outcome setup. If strong LOSR-GMNL is weakened to an $\epsilon$-close condition, we demonstrate that \(n|(n-1)\) resource separation can be robustly observed, theoretically allowing experiments observationally indistinguishable from true pseudotelepathy. By lifting the powerful all-or-nothing nonlocality of pseudotelepathic correlations to a stronger genuinely-multipartite form across a broad class of minimal scenarios, these results lay the groundwork for a deeper understanding of this phenomenon and its potential applications in near-future multipartite quantum networks.  

\begin{acknowledgments}
The formulation of ideas, proving/checking of all results, and all written text are due solely to the authors. The only use of AI was in designing the figures. The authors acknowledge helpful conversations with Arkaprabha Ghosal, Mehdi Mhalla and Elie Wolfe. This work was partially supported by NSF Award No.~2547582.
\end{acknowledgments}

\bibliography{metabib}

\onecolumngrid
\appendix

\section{The Game, and the Resource 
Models}\label{a:5CycleGame}

\subsection{A Perfect Quantum Strategy}\label{s:qstrat}
For the predicates defined in~\eqref{e:5game} in the main text, we specify the quantum state shared by the five players and the local measurements that are applied by the players to their respective systems to win the associated 5-cycle game with probability \(1\), thereby defining the quantum behavior including the probabilities beyond the six winning correlations. The shared state is the normalised zero-codeword state used by DiVincenzo and Peres in~\cite{DiVincenzoPeres97}:
\begin{align*}
|\psi\rangle = \frac{1}{4}\big(&-|00000\rangle - |11000\rangle - |01100\rangle - |00110\rangle - |00011\rangle - |10001\rangle + |10010\rangle + |10100\rangle \\
&+ |01001\rangle + |01010\rangle + |00101\rangle + |11110\rangle + |11101\rangle + |11011\rangle + |10111\rangle + |01111\rangle\big).
\end{align*}
We label the five qubits as \(1,2,3,4,5\), corresponding to the players \(\mathsf{A}_1, \mathsf{A}_2, \mathsf{A}_3, \mathsf{A}_4, \mathsf{A}_5\). These correspond to the parties \(\mathsf{A}, \mathsf{B}, \mathsf{C}, \mathsf{D}, \mathsf{E}\) in the main text; using instead numbered parties makes it easier to notate some manipulations in the SM exploiting the cyclic symmetry among parties in both the above state and the predicates in~\eqref{e:5game}. The qubits are indexed modulo \(5\). On receiving input \(X\) or \(Z\), each player measures the corresponding Pauli observable \(\sigma_{\mathsf X}\) or \(\sigma_{\mathsf Z}\) on their own qubit. The measurement result is recorded as a bit: Eigenvalues \(+1\) and \(-1\) are recorded as bit values \(0\) and \(1\), respectively, i.e., an output bit \(a\) represents the measurement result \((-1)^a\). In particular, an odd parity of output bits corresponds to a product of the measurement results---the eigenvalues---equal to \(-1\), while an even parity corresponds to a product equal to \(+1\).

To see why this strategy wins, write:
\begin{equation}
\mathsf{X}_i\coloneqq\sigma^{(i)}_{\mathsf X},\,\,\,\mathsf{Z}_i\coloneqq\sigma^{(i)}_{\mathsf Z},\,\,\,H_i\coloneqq\mathsf{X}_i\mathsf{Z}_{i+1}\mathsf{X}_{i+2},
\end{equation}
with identity understood on all unused qubits. For instance, \(H_1 = \mathsf{X}_1\otimes\mathsf{Z}_2\otimes\mathsf{X}_3\otimes I\otimes I\), \(H_2=I\otimes\mathsf{X}_1\otimes\mathsf{Z}_2\otimes\mathsf{X}_3\otimes I\), and so on. The shared state satisfies \(H_i|\psi\rangle=-|\psi\rangle\) for every \(i\in\{1,2,\ldots,5\}\). This can be checked from the fact that \(\sigma_{\mathsf X}|a\rangle=|a\oplus 1\rangle\) and \(\sigma_{\mathsf Z}|a\rangle=(-1)^a|a\rangle\). And so, for example, 
\begin{equation*}
H_1|abcde\rangle=(-1)^b|a\oplus 1\rangle|b\rangle|c\oplus 1\rangle|d\rangle|e\rangle.
\end{equation*}
Applying this rule to the above state gives \(H_1|\psi\rangle=-|\psi\rangle\). The other four equations follow from cyclic permutation, under which the state is invariant. The eigenvalue equations can be understood operationally as follows: For the three-player predicates, when the players \(\mathsf{A}_i, \mathsf{A}_{i+1}, \mathsf{A}_{i+2}\) measure \(\sigma_{\mathsf X}, \sigma_{\mathsf Z}, \sigma_{\mathsf X}\), respectively, the product of their measurement results is \(-1\) with certainty because \(|\psi\rangle\) is an eigenvector of \(H_i\) corresponding to eigenvalue \(-1\), for \(i\in\{1,\ldots,5\}\). Consequently, \((-1)^{A^{(i)}_X+A^{(i+1)}_Z+A^{(i+2)}_X}=-1\), i.e., \(A^{(i)}_X\oplus A^{(i+1)}_Z\oplus A^{(i+2)}_X = 1\).

The satisfaction of the all-\(Z\) predicate can be checked directly, but it also follows from the identity \(\prod_{i=1}^5 H_i = -(\mathsf{Z}_1\otimes\mathsf{Z}_2\otimes\mathsf{Z}_3\otimes\mathsf{Z}_4\otimes\mathsf{Z}_5)\). Since \(\prod_{i=1}^5 H_i|\psi\rangle=(-1)^5|\psi\rangle=-|\psi\rangle\), we have
\begin{equation}
 (\mathsf{Z}_1\otimes\mathsf{Z}_2\otimes\mathsf{Z}_3\otimes\mathsf{Z}_4\otimes\mathsf{Z}_5)|\psi\rangle = |\psi\rangle.
\end{equation}
Consequently, \((-1)^{A^{(1)}_Z+\cdots+A^{(5)}_Z}=1\), i.e., \(A^{(1)}_Z\oplus\cdots\oplus A^{(5)}_Z=0\). The full five-party distribution is thus given by local projective measurements of the Pauli observables \(\sigma_{\mathsf X}\) and \(\sigma_{\mathsf Z}\) on the shared state \(|\psi\rangle\). Each party's setting specifies which observable is measured, and the binary outcomes \(0\) and \(1\) label its eigenvalues \(+1\) and \(-1\), respectively.

\subsection{Proof of Proposition 1}\label{s:prop1proof}
\begin{prop*}
The maximum winning probability of the \(5\)-cycle game when two players respond deterministically on input \(X\) is \(5/6\).
\end{prop*}
The proof of this proposition relies on a simple fact stated and proved in the following lemma.
\begin{lemma}\label{l:Lemma_for_Prop1}
For binary-valued random variables \(R,S,T\) and any two bits \(\gamma,\delta\), we have: 
\begin{equation}
\mathbb{P}(R\oplus T = \gamma\oplus\delta\oplus 1)\le\mathbb{P}(R\oplus S=\gamma) + \mathbb{P}(S\oplus T=\delta).
\end{equation}
\end{lemma}
\begin{proof}
If the events \(A\coloneqq\{R\oplus S=\gamma\oplus 1\}\) and \(B\coloneqq\{S\oplus T=\delta\oplus 1\}\) both occur, then the event \(C\coloneqq\{R\oplus T=\gamma\oplus\delta\}\) occurs as well. So \(A\cap B\subseteq C\). But the reverse does not necessarily hold, since \(R=\gamma, T=\delta\) with \(S=0\) satisfies \(C\) but not \(A\cap B\). Taking the complement of both sides reverses the inclusion: \(C^{\mathrm{c}}\subseteq(A\cap B)^{\mathrm{c}}=A^{\mathrm{c}}\cup B^{\mathrm{c}}\). Then applying the union bound on probabilities yields \(\mathbb{P}(C^{\mathrm{c}})\le\mathbb{P}(A^{\mathrm{c}})+\mathbb{P}(B^{\mathrm{c}})\), which is the required inequality in the statement of the lemma.
\end{proof}
We now proceed with the proof of the proposition.
\begin{proof}
Consider a strategy where two players respond deterministically on input \(X\). Define $W_i$ to be the event corresponding to satisfaction of the predicate defined in~\eqref{e:5game} for which the $i$-th player is the (sole) player measuring $Z$; we use $W_f$ to distinguish the final all-$Z$ predicate from the cyclically equivalent first five. We prove \(\mathcal{S}:=\mathbb{P}(W_f)+\sum_{i=1}^{5}\mathbb{P}(W_i)\le 5\), which implies the average predicate-winning probability $\omega$ satisfies \(\omega\le 5/6\). Let \(e_i\coloneqq\mathbb{P}(W^{\mathrm{c}}_i)\), for \(1\le i\le 5\), and \(e_f\coloneqq\mathbb{P}(W^{\mathrm{c}}_f)\). Let the sum of the failure probabilities of the cyclic three-player predicates be \(e_T=\sum_{i=1}^{5}e_i\). The total probability of failure is then \(e=e_f+e_T\). We first prove the result for a pair of adjacent players \((\mathsf{A}_1,\mathsf{A}_2)\), following which we prove it for a non-adjacent pair \((\mathsf{A}_1,\mathsf{A}_3)\). Due to the cyclic symmetry of the predicates, this covers all distinct cases.

Suppose \(A^{(1)}_X=\alpha\) and \(A^{(2)}_X=\beta\) deterministically. Then \(\mathbb{P}(W_2)=\mathbb{P}(A^{(2)}_Z\oplus A^{(3)}_X=1\oplus\alpha)=1-e_2\) and \(\mathbb{P}(W_4)=\mathbb{P}(A^{(3)}_X\oplus A^{(4)}_Z\oplus A^{(5)}_X=1)=1-e_4\), or $\mathbb{P}(W_2^{\mathrm{c}})=\mathbb{P}(A^{(2)}_Z\oplus A^{(3)}_X=\alpha)=e_2$ and \(\mathbb{P}(W_4^{\mathrm{c}})=\mathbb{P}(A^{(3)}_X\oplus A^{(4)}_Z\oplus A^{(5)}_X=0)=e_4\). Using Lemma~\ref{l:Lemma_for_Prop1} (with $R=A^{(2)}_Z$, $S=A^{(3)}_X$, and $T=A^{(4)}_Z\oplus A^{(5)}_X$) we have 
\begin{equation*}
\mathbb{P}(A^{(2)}_Z\oplus A^{(4)}_Z\oplus A^{(5)}_X=1\oplus\alpha)\le\mathbb{P}(A^{(2)}_Z\oplus A^{(3)}_X=\alpha)+\mathbb{P}(A^{(3)}_X\oplus A^{(4)}_Z\oplus A^{(5)}_X=0)=e_2+e_4.
\end{equation*}
Again, because of the fixed outcomes, we have \(\mathbb{P}(W_3^{\mathrm{c}})=\mathbb{P}(A^{(3)}_Z\oplus A^{(4)}_X=\beta)=e_3\) and \(\mathbb{P}(W_5^{\mathrm{c}})=\mathbb{P}(A^{(4)}_X\oplus A^{(5)}_Z=\alpha)=e_5\), and following a similar usage of Lemma~\ref{l:Lemma_for_Prop1}, we get \(\mathbb{P}(A^{(3)}_Z\oplus A^{(5)}_Z=\alpha\oplus\beta\oplus 1)\le e_3 + e_5\). For the remainder of the proof, ``combining'' two parity probabilities involves the usage of Lemma~\ref{l:Lemma_for_Prop1} in this manner. Combining \(\mathbb{P}(W^{\mathrm c}_1)=\mathbb{P}(A^{(5)}_X\oplus A^{(1)}_Z=\beta)=e_1\) and \(\mathbb{P}(A^{(2)}_Z\oplus A^{(4)}_Z\oplus A^{(5)}_X=1\oplus\alpha)\le e_2+e_4\), we get \(\mathbb{P}(A^{(1)}_Z\oplus A^{(2)}_Z\oplus A^{(4)}_Z=\alpha\oplus\beta)\le e_1+e_2+e_4\). Combining this with the expression \(\mathbb{P}(A^{(3)}_Z\oplus A^{(5)}_Z=\alpha\oplus\beta\oplus 1)\le e_3 + e_5\) derived earlier (while using the constant random variable $S=0$ in the application of Lemma~\ref{l:Lemma_for_Prop1}), we get
\begin{equation*}
\mathbb{P}\Big(\bigoplus_{i=1}^{5}A^{(i)}_Z=0\Big)\le \sum_{i=1}^{5}e_i \Rightarrow \mathbb{P}\Big(\bigoplus_{i=1}^{5}A^{(i)}_Z=1\Big)\ge 1-\sum_{i=1}^{5}e_i=1-e_T,
\end{equation*}
comparing which with \(\mathbb{P}(W_f)=\mathbb{P}(\bigoplus_{i=1}^{5}A^{(i)}_Z=0)=1-e_f\), we find that the target parities are opposite. Therefore, summing the two expressions gives us:
\begin{equation*}
1 = \mathbb{P}\Big(\bigoplus_{i=1}^{5}A^{(i)}_Z=0\Big) + \mathbb{P}\Big(\bigoplus_{i=1}^{5}A^{(i)}_Z=1\Big) \ge 1 - e_f + 1-e_T=2-e,
\end{equation*}
implying that \(e\ge 1\). This shows that \(\mathcal{S}=6-e\le 5\), equivalently, \(\omega\le 5/6\).

Now assume that \(A^{(1)}_X=\gamma\) and \(A^{(3)}_X=\delta\). Since \(\mathbb{P}(W^{\mathrm c}_2)=e_2\), we have \(\mathbb{P}(A^{(2)}_Z=\gamma\oplus\delta)=e_2\). Combining \(\mathbb{P}(W^{\mathrm c}_4)=\mathbb{P}(A^{(4)}_Z\oplus A^{(5)}_X=\delta)=e_4\) and \(\mathbb{P}(W^{\mathrm c}_1)=\mathbb{P}(A^{(1)}_Z\oplus A^{(2)}_X\oplus A^{(5)}_X=0)=e_1\): 
\begin{equation}\label{eq:2FixedPlayers_1}
\mathbb{P}(A^{(1)}_Z\oplus A^{(2)}_X\oplus A^{(4)}_Z=\delta\oplus 1)\le e_1 + e_4.
\end{equation}
Combining \(\mathbb{P}(W^{\mathrm{c}}_3)=\mathbb{P}(A^{(2)}_X\oplus A^{(3)}_Z\oplus A^{(4)}_X=0)=e_3\) and \(\mathbb{P}(W^{\mathrm{c}}_5)=\mathbb{P}(A^{(4)}_X\oplus A^{(5)}_Z=\gamma)=e_5\):
\begin{equation}\label{eq:2FixedPlayers_2}
\mathbb{P}(A^{(2)}_X\oplus A^{(3)}_Z\oplus A^{(5)}_Z=\gamma\oplus 1)\le e_3 + e_5.
\end{equation}
Combining~\eqref{eq:2FixedPlayers_1} and~\eqref{eq:2FixedPlayers_2}, we have \(\mathbb{P}(A^{(1)}_Z\oplus A^{(3)}_Z\oplus A^{(4)}_Z\oplus A^{(5)}_Z=\gamma\oplus\delta\oplus1)\le e_1+e_3+e_4+e_5\), which in combination with \(\mathbb{P}(A^{(2)}_Z=\gamma\oplus\delta)=e_2\) gives:
\begin{equation}
\mathbb{P}\Big(\bigoplus_{i=1}^{5}A^{(i)}_Z=0\Big)\le\sum_{i=1}^{5}e_i \Rightarrow \mathbb{P}\Big(\bigoplus_{i=1}^{5}A^{(i)}_Z=1\Big)\ge 1-\sum_{i=1}^{5}e_i=1-e_T.
\end{equation}
As in the earlier case, comparing this with \(\mathbb{P}(W_f)=\mathbb{P}(\bigoplus_{i=1}^{5}A^{(i)}_Z=0)=1-e_f\), we find that the target parities are opposite. Summing the expressions gives:
\begin{equation}
1 = \mathbb{P}\Big(\bigoplus_{i=1}^{5}A^{(i)}_Z=0\Big) + \mathbb{P}\Big(\bigoplus_{i=1}^{5}A^{(i)}_Z=1\Big) \ge 1-e_f + 1-e_T = 2-e,
\end{equation}
implying that \(e\ge 1\), which proves the necessary result following the same arguments as in the previous case.
\end{proof}

\subsection{Proof of Proposition 2}\label{s:prop2proof}
The main text describes an exact obstruction (by means of inflation) to the pairwise-resource model winning the \(5\)-cycle game, but this does not establish a bound on the winning probability. Here we provide this by proving Proposition~\ref{prop:ImprovedBndBipartiteRes}. Prior to that, we state the underlying resource model explicitly which identifies the independences the inflation proof may use.

\textit{Pairwise-Resource Model}: For each unordered pair of players \((\mathsf{A}_i,\mathsf{A}_j)\) in an \(n\)-player setup, let \(R_{ij}\) be a bipartite arbitrary no-signaling resource. Distinct resources are independent conditional on globally shared classical randomness (a random variable) \(\lambda\). Each player may apply an arbitrary local wiring to all incident resources and to \(\lambda\), and/or apply quantum entangled measurements (or even generalized measurements under alternate theories, so long as they are compatible with the two assumptions below), but the players are not allowed to communicate once they receive their inputs.

Inflation duplicates the resources and parties, and reconnects the copies so that every observable duplicate party is connected to the same local resource types as its original. Two assumptions of an inflation network are:
\begin{enumerate}
\item Any set of copied parties with the same resource ancestry as an original set has the same marginal distribution.
\item Copies with disjoint resource ancestry are independent conditional on the shared randomness \(\lambda\).
\end{enumerate}
These are the \textit{marginal isomorphism} and \textit{ancestral independence} constraints, respectively~\cite{wolfe:2019}; they are satisfied by quantum mechanics as well as other theories including wirings of PR boxes~\cite{bierhorst:2024}. No-signaling is needed when unobserved copies are discarded: their inputs and local wirings cannot alter the retained marginal. The inflation in Figure~\ref{fig:Two-copyInflation} contains precisely the marginals needed below. For example, \((\mathsf{A}',\mathsf{B}',\mathsf{C}')\), \((\mathsf{D}',\mathsf{E}',\mathsf{A}'')\) and \((\mathsf{B}',\mathsf{C}',\mathsf{D}', \mathsf{E}')\) reproduce their corresponding original marginals, whereas \(\mathsf{A}'\) and \(\mathsf{A}''\) share no resource. The same construction applies after a cyclic relabeling of the five players. 

In the main text, we demonstrated that no pairwise-resource-network-strategy can win the game with certainty. An exact contradiction, however, does not by itself exclude strategies whose winning probability can get arbitrarily close to one. Here we show that the same inflation has a quantitative content: near-perfect parities make each output close to deterministic, and the deterministic parity contradiction then forces a non-zero total failure probability.

For our proof of Proposition~\ref{p:RobustBound}, we will use the following lemma.

\begin{lemma}\label{l:PerfectParityAcrossCondIndSources_Robust}
Let \(U,V\) be binary random variables that are independent and identically distributed. If \(\prbP(U\neq V)\le\eta\), then there exists \(u_\ast\in\{0,1\}\) such that \(\prbP(U\neq u_\ast)\le\eta\). 
\end{lemma}
\begin{proof}
Let \(p\coloneqq \prbP(U=0)\). Since \(U\text{ and }V\) are independent and identically distributed:
\begin{equation*}
\eta\ge \prbP(U\neq V) = p(1-p) + (1-p)p = 2p(1-p)\ge\delta,
\end{equation*}
where \(\delta\coloneqq\min\{p,1-p\}\). Choose \(u_\ast=0\) with \(p\ge1/2\), and \(u_\ast=1\) otherwise. Then \(\prbP(U\neq u_\ast)=\delta\le\eta\).
\end{proof}

\begin{prop*}\label{p:RobustBound}The average winning probability of the game for every behavior in the pairwise-resource network satisfies \(\omega\le 1-\frac{5}{126}\).
\end{prop*}
\begin{proof}
All probabilities in the proof are conditional on a fixed value of the globally shared classical randomness \(\lambda\). The derived bound on winning probability holds for this fixed value, but as the bound is linear in probabilities, it will also hold for the complete behavior after averaging over the shared randomness. Hence we omit explicit mention of \(\lambda\) for the rest of the proof. The players are indexed cyclically with indices modulo \(5\). Recall we notate the events corresponding to winning the (cyclic) three-player predicates and the all-\(Z\)-input predicate as 
\begin{equation}\label{eq:Predicates_BP}
\begin{split}
W_i&\coloneqq A^{(i-1)}_X\oplus A^{(i)}_Z\oplus A^{(i+1)}_X=1,\,1\le i\le 5,\\
W_f&\coloneqq A^{(1)}_Z\oplus A^{(2)}_Z\oplus A^{(3)}_Z\oplus A^{(4)}_Z\oplus A^{(5)}_Z = 0.
\end{split}
\end{equation}
For a given strategy \(\mathfrak S\), the probability of the event \(\{\mathfrak{S}\text{ wins on } Q_k\}\) is denoted as the probability \(\mathbb{P}_{\strat}(W_k)\) of the corresponding winning predicate given \(\mathfrak S\). Denote the probability of the event \(\{\mathfrak{S}\text{ loses on }Q_k\}\) as \(e_k\coloneqq \mathbb{P}_\strat(W^{\mathrm c}_k)\), for \(1\le k\le 5\) and \(e_f\coloneqq\mathbb{P}_\strat(W^{\mathrm c}_f)\). Also, let the sum of the loss probabilities on the five cyclic questions (those where there are exactly three active players) be \(e_T = \sum_{k=1}^5 e_k\). The total probability of losing the game is then \(e=e_T + e_f\).

We first establish approximate determinism for the players choosing \(X\) as their input. Fix \(u\in\{1,\ldots,5\}\). The two three-player predicates involving \(A^{(u)}_X\) are \(A^{(u)}_X\oplus A^{(u+1)}_Z\oplus A^{(u+2)}_X=1\) and \(A^{(u+3)}_X\oplus A^{(u+4)}_Z\oplus A^{(u)}_X=1\), the respective failure probabilities of which are \(e_{u+1}\) and \(e_{u+4}\). In the global input context \((X,Z,X,X,Z)\) for the parties \((\mathsf{A}_u,\mathsf{A}_{u+1},\mathsf{A}_{u+2},\mathsf{A}_{u+3},\mathsf{A}_{u+4})\), the two predicates can hold simultaneously due to no-signaling, and so \(A^{(u+1)}_Z\oplus A^{(u+2)}_X\oplus A^{(u+3)}_X\oplus A^{(u+4)}_Z = 0\) holds, except with probability at most \(e_{u+1} + e_{u+4}\), which follows from the union bound for probabilities--- explicitly, an application of Lemma \ref{l:Lemma_for_Prop1} with $R=A^{(u+1)}_Z\oplus A^{(u+2)}_X$, $S=A^{(u)}_X$, and $T=A^{(u+3)}_X\oplus A^{(u+4)}_Z$.

Moving to the inflated network as depicted on the right side of Figure~\ref{fig:Two-copyInflation}, the relevant inflation marginal contains:
\begin{equation*}
A'^{(u)}_X, A'^{(u+1)}_Z, A'^{(u+2)}_X, A'^{(u+3)}_X, A'^{(u+4)}_Z, A''^{(u)}_X,
\end{equation*}
where \(A'^{(k)}_s,A''^{(k)}_s\) are the source-disjoint copies of \(A^{(k)}_s\), for \(k\in\{1,\ldots,5\}\) and \(s\in\{X,Z\}\). By the inflation isomorphism constraint, the inflation marginal satisfies:
\begin{equation}
\begin{split}
\mathbb{P}(A'^{(u)}_X\oplus A'^{(u+1)}_Z\oplus A'^{(u+2)}_X\neq 1) &= e_{u+1},\\
\mathbb{P}(A'^{(u+3)}_X\oplus A'^{(u+4)}_Z\oplus A''^{(u)}_X\neq 1) &= e_{u+4},\\
\mathbb{P}(A'^{(u+1)}_Z\oplus A'^{(u+2)}_X\oplus A'^{(u+3)}_X\oplus A'^{(u+4)}_Z\neq 0) &\le e_{u+1}+e_{u+4}.
\end{split}
\end{equation}
If the three ``good events'', namely, 
\begin{equation*}
\{A'^{(u)}_X\oplus A'^{(u+1)}_Z\oplus A'^{(u+2)}_X = 1\}, \{A'^{(u+3)}_X\oplus A'^{(u+4)}_Z\oplus A''^{(u)}_X = 1\}, \{A'^{(u+1)}_Z\oplus A'^{(u+2)}_X\oplus A'^{(u+3)}_X\oplus A'^{(u+4)}_Z = 0\},
\end{equation*}
occur, then \(A'^{(u)}_X = A''^{(u)}_X\), so the event \(\{A'^{(u)}_X = A''^{(u)}_X\}\) implies the the negation of at least one of the above events, and so by the union bound
\begin{equation}
\mathbb{P}(A'^{(u)}_X \neq A''^{(u)}_X)\le 2(e_{u+1}+e_{u+4}).
\end{equation}
Since \(A'^{(u)}_X\) and \(A''^{(u)}_X\) are source-disjoint, they are independent, and since they are locally isomorphic copies of \(A^{(u)}_X\), they are identically distributed. By Lemma~\ref{l:PerfectParityAcrossCondIndSources_Robust}, there exists some \(\alpha_u\in\{0,1\}\) such that \(\mathbb{P}(A^{(u)}_X\neq\alpha_u)\le 2(e_{u+1}+e_{u+4})\). Define \(d^X_u\coloneqq\mathbb{P}(A^{(u)}_X\neq\alpha_u)\) for the choice of $\alpha_u$ provided by the lemma for party $u$. Then \(d^X_u\le 2(e_{u+1}+e_{u+4})\). Since this holds for any \(u\in\{1,\ldots,5\}\), summing over \(u\) gives:
\begin{equation}\label{eq:SumMismatchProb_Bnd}
\sum_{u=1}^{5}d^X_u \le 2\sum_{u=1}^{5}(e_{u+1}+e_{u+4})=4e_T,
\end{equation}
where equality follows from the fact that \(i+1\mapsto i+4\) is a permutation of the indices, and so \(\sum_{u=1}^{5}(e_{u+1}+e_{u+4})=2e_T\).  

Consider now a strategy \(\mathfrak{S}'\) where any two players \(\mathsf{A}_u,\mathsf{A}_v\), for \(u\neq v\), behave deterministically on one of their inputs, say \(X\). The way we ``construct'' such a \(\mathfrak{S}'\) is to keep its internal protocol exactly the same as \(\mathfrak{S}\), but with the exception that players \(\mathsf{A}_u\) and \(\mathsf{A}_v\), instead of declaring their outputs as a function of the outcomes of all bipartite no-signaling resources and \(X\), now choose to ``overwrite'' it deterministically with \(A^{(u)}_X=\alpha_u\) and \(A^{(v)}_X=\alpha_v\) obtained via Lemma~\ref{l:PerfectParityAcrossCondIndSources_Robust} as described above. Fix a question \(Q_i\). As shown in Table~\ref{tab:BPGame_Questions_ActiveSet}, the \textit{active set} \(I(Q_i)\) consists of the outputs of players involved in the winning predicate \(W_i\) corresponding to \(Q_i\).
\begin{table}[!htb]
\caption{\label{tab:BPGame_Questions_ActiveSet}The six questions defining the five-player parity game. For each cyclic question \(Q_i\), \(i\in\{1,\ldots,5\}\), player \(\mathsf{A}_i\) receives input \(Z\), its two cyclic neighbours receive input \(X\), and the remaining two players receive an inactive input denoted \(\perp\), whose corresponding outputs do not enter the winning predicate. The set \(I(Q_i)\) lists the outputs of the three active players, and \(W_i\) requires them to have odd parity. For the all-\(Z\) question \(Q_a\), every player in \(\mathcal P=\{\mathsf{A}_1,\ldots,\mathsf{A}_5\}\) is active, and the players win when the parity of all five outputs is even. All player indices are understood cyclically modulo \(5\), and \(\oplus\) denotes addition modulo \(2\). The winning predicates are \(W_i\coloneqq \bigoplus_{v\in I(Q_i)}v=1\), for \(1\le i\le 5\), and \(W_f\coloneqq\bigoplus_{i=1}^{6}A^{(i)}_Z=0\).}
\centering
\begin{tabular}{ccccccc}
\toprule
\multirow{2}{*}{\(Q_i\)} & \multicolumn{5}{c}{Player Inputs} & \multirow{2}{*}{\(I(Q_i)\)} \\
\cmidrule{2-6}
 & \(\mathsf{A}_1\) & \(\mathsf{A}_2\) & \(\mathsf{A}_3\) & \(\mathsf{A}_4\) & \(\mathsf{A}_5\) & \\
\midrule
\(Q_1\) & \(Z\) & \(X\) & \(\perp\) & \(\perp\) & \(X\) & \(\{A^{(5)}_{X},A^{(1)}_{Z},A^{(2)}_{X}\}\) \\
\(Q_2\) & \(X\) & \(Z\) & \(X\) & \(\perp\) & \(\perp\) & \(\{A^{(1)}_{X},A^{(2)}_{Z},A^{(3)}_{X}\}\) \\
\(Q_3\) & \(\perp\) & \(X\) & \(Z\) & \(X\) & \(\perp\) & \(\{A^{(2)}_{X},A^{(3)}_{Z},A^{(4)}_{X}\}\) \\
\(Q_4\) & \(\perp\) & \(\perp\) & \(X\) & \(Z\) & \(X\) & \(\{A^{(3)}_{X},A^{(4)}_{Z},A^{(5)}_{X}\}\) \\
\(Q_5\) & \(X\) & \(\perp\) & \(\perp\) & \(X\) & \(Z\) & \(\{A^{(4)}_{Z},A^{(5)}_{X},A^{(1)}_{X}\}\) \\
\(Q_a\) & \(Z\) & \(Z\) & \(Z\) & \(Z\) & \(Z\) & \(\{A^{(1)}_Z,A^{(2)}_Z,A^{(3)}_Z,A^{(4)}_Z,A^{(5)}_Z\}\) \\
\bottomrule
\end{tabular}
\end{table}
For any fixed \(Q_i\), the following possibilities arise: (1) neither \(A^{(u)}_X\) nor \(A^{(v)}_X\) belongs to the active set, or (2) either \(A^{(u)}_X\) or \(A^{(v)}_X\) (or both) belongs to the active set. These possibilities can be formalized as the events \(\{A^{(u)}_X\notin I(Q_i)\}\cap\{A^{(v)}_X\notin I(Q_i)\}\), for the former, and \(\{A^{(u)}_X\in I(Q_i)\}\cup\{A^{(v)}_X\in I(Q_i)\}\), for the latter. When neither \(A^{(u)}_X\) nor \(A^{(v)}_X\) belongs to the active set, the probability of the event \(\{\mathfrak{S}\text{ wins on }Q_i\)\} is the same as that of \(\{\mathfrak{S}'\text{ wins on }Q_i\)\}, i.e., \(\mathbb{P}_\strat(W_i)=\mathbb{P}_{\strat'}(W_i)\), since the two strategies \(\mathfrak{S}\) and \(\mathfrak{S}'\) are identical in this case. Now suppose that \(A^{(u)}_X\in I(Q_i)\), but \(A^{(v)}_X\notin I(Q_i)\). If \(A^{(u)}_X=\alpha_u\) according to the original strategy \(\strat\), then the strategy \(\strat'\) does not contribute anything relevant to the winning predicate of \(Q_i\). (This is because, by construction, \(\strat'\) differs from \(\strat\) only in the replacement of the outcome \(A^{(u)}_X\) to \(\alpha_u\)). Then whenever \(\strat\) wins on \(Q_i\) but \(\strat'\) does not, it must be that \(A^{(u)}_X\neq \alpha_u\) in the original strategy. Equivalently, we have:
\begin{equation}\label{eq:Au_X_belongsto_Qi}
\{\strat\text{ wins on }Q_i\}\subseteq\{\strat'\text{ wins on }Q_i\}\cup\{A^{(u)}_X\neq \alpha_u\text{ on }\strat\}\Rightarrow \mathbb{P}_\strat(W_i) \le \mathbb{P}_{\strat'}(W_i)+\mathbb{P}_{\strat}(A^{(u)}_X\neq\alpha_u).
\end{equation}
Next, suppose \(A^{(v)}_X\in I(Q_i)\) but \(A^{(u)}_X\notin I(Q_i)\). Using the same reasoning as for the earlier case, we have:
\begin{equation}\label{eq:Av_X_belongsto_Qi}
\mathbb{P}_{\strat}(W_i) \le \mathbb{P}_{\strat'}(W_i)+\mathbb{P}_{\strat}(A^{(v)}_X\neq\alpha_v).
\end{equation}
For the case \(A^{(u)}_X,A^{(v)}_X\in I(Q_i)\), we can adapt the above reasoning and obtain:
\begin{equation}\label{eq:Au_X_Av_X_belongsto_Qi}
\begin{split}
\{\strat\text{ wins on }Q_i\}&\subseteq\{\strat'\text{ wins on }Q_i\}\cup\{A^{(u)}_X\neq\alpha_u\text{ on }\strat\}\cup\{A^{(v)}_X\neq\alpha_v\text{ on }\strat\}\\
\Rightarrow \mathbb{P}_{\strat}(W_i)&\le\mathbb{P}_{\strat'}(W_i)+\mathbb{P}_{\strat}(A^{(u)}_X\neq\alpha_u)+\mathbb{P}_{\strat}(A^{(v)}_X\neq\alpha_v).
\end{split}
\end{equation}
All the cases in~\eqref{eq:Au_X_belongsto_Qi},~\eqref{eq:Av_X_belongsto_Qi}, and~\eqref{eq:Au_X_Av_X_belongsto_Qi} together with the case \(A^{(u)}_X\notin I(Q_i),A^{(v)}_X\notin I(Q_i)\) can be expressed in a unified manner as shown below, where ``\([\![\cdots]\!]\)" evaluates to 1/0 if the enclosed expression is true/false:
\begin{equation}\label{eq:ProbIneq_S_S'}
\mathbb{P}_{\strat}(W_i)\le\mathbb{P}_{\strat'}(W_i)+[\![A^{(u)}_X\in I(Q_i)]\!]\mathbb{P}_{\strat}(A^{(u)}_X\neq\alpha_u)+[\![A^{(v)}_X\in I(Q_i)]\!]\mathbb{P}_{\strat}(A^{(v)}_X\neq\alpha_v),
\end{equation}
which holds as an equality if \(A^{(u)}_X\notin I(Q_i)\) and \(A^{(v)}_X\notin I(Q_i)\). Recalling the definitions \(d^X_u=\mathbb{P}_{\strat}(A^{(u)}_X\neq\alpha_u)\) and \(d^X_v=\mathbb{P}_{\strat}(A^{(v)}_X\neq\alpha_v)\) for the mismatch probabilities, the inequality in~\eqref{eq:ProbIneq_S_S'} can be expressed equivalently as:
\begin{equation}\label{eq:ProbIneq_S_S'_equiv}
\mathbb{P}_{\strat}(W_i)\le\mathbb{P}_{\strat'}(W_i)+[\![A^{(u)}_X\in I(Q_i)]\!]d^X_u+[\![A^{(v)}_X\in I(Q_i)]\!]d^X_v.   
\end{equation}
Then summing over \(i\), and noting that \(\sum_{i=1}^6\mathbb{P}_{\strat}(W_i)=\mathcal{S}_{\strat}\), and similarly for \(\strat'\), we get:
\begin{equation}\label{eq:BP_value_bound}
\mathcal{S}_{\strat}\le \mathcal{S}_{\strat'}+\sum_{i=1}^{6}[\![A^{(u)}_X\in I(Q_i)]\!]d^X_u + \sum_{i=1}^{6}[\![A^{(v)}_X\in I(Q_i)]\!]d^X_v \le \mathcal{S}_{\strat'}+2(d^X_u + d^X_v)\le 5 + 2(d^X_u + d^X_v),
\end{equation}
where the second inequality follows from the fact that for any given player \(\mathsf{A}_k\), the \(X\)-input outcome \(A^{(k)}_X\) appears in exactly two out of the five cyclic predicates, and the third inequality follows from Proposition~\ref{p:2DetPlayers} as the strategy $\mathcal{S}_{\strat'}$ is deterministic on two parties' $X$ outcomes. 

We now prove that there is a choice of $u$ and $v$ for which \(2(d^X_u + d^X_v)\le \frac{16}{5}e_T\), where \(e_T=\sum_{i=1}^{5}e_i\) is the sum of the failure probabilities of the five three-player predicates under strategy \(\mathfrak{S}\). This result combined with~\eqref{eq:BP_value_bound} will establish the required winning probability bound. There are \(\binom{5}{2}=10\) unordered pairs of players. Moreover, each mismatch probability \(d^X_u\) occurs in exactly four of the ten pair-sums \(d^X_u+d^X_v\), for \(u<v\). Therefore,
\begin{equation}
\frac{1}{10}\sum_{1\le u<v\le 5}(d^X_u + d^X_v) = \frac{4}{10}\sum_{i=1}^{5}d^X_i\le\frac{4}{10}(4e_T)=\frac{8}{5}e_T,
\end{equation}
where the inequality follows from the bound on the sum of mismatch probabilities in~\eqref{eq:SumMismatchProb_Bnd}. At least one pair-sum must be less than or equal to the average. Consequently, there exists a pair \(\{u,v\}\) such that
\begin{equation}
d^X_u + d^X_v \le \frac{1}{10}\sum_{1\le u<v\le 5}(d^X_u + d^X_v) \le \frac{8}{5}e_T, 
\end{equation}
which implies that \(2(d^X_u + d^X_v)\le\frac{16}{5}e_T\). Notice that this need not hold for a pair \(\{u,v\}\) fixed beforehand. However, since every pair is either adjacent or non-adjacent, and our two-players-deterministic-on-\(X\) obstruction covers both cases, the minimizing pair can be used. Therefore, combining this with~\eqref{eq:BP_value_bound} we get:
\begin{equation}
\mathcal{S}_{\strat}\le 5 + \frac{16}{5}e_T\quad\Rightarrow\quad 6-(e_T + e_f)\le 5 + \frac{16}{5}e_T \quad\Rightarrow\quad \frac{21}{5}e_T + e_f \ge 1.    
\end{equation}
Since \(0\le e_f\), it follows that \(\frac{21}{5}(e_T+e_f)\ge\frac{21}{5}e_T+e_f\ge1\), which implies \(e=e_T + e_f \ge \frac{5}{21}\). Since each question is sampled uniformly, this implies \(\frac{e}{6}\ge\frac{5}{126}\). Equivalently, \(\omega=\frac{1}{6}\mathcal{S}_{\strat}=\frac{1}{6}(6-e)\le 1-\frac{5}{126}\).
\end{proof}

\section{A Four-Player No-Signaling Simulation}\label{s:4-1split}
The preceding results rule out bipartite no-signaling resources, but they do not show that five-player irreducibility is necessary. To locate the next point in the resource hierarchy, we freeze one player’s response and ask whether a single four-way no-signaling resource can supply the remaining correlations.

\subsection{One Deterministic Player}
Enforcing one party, say \(\mathsf{A_1}\), to respond deterministically on its inputs, by prescribing \(A_{X}^{(1)}=r\) and \(A_{Z}^{(1)}=1\oplus s\), for \(r,s\in\{0,1\}\), transforms the six predicates in~\eqref{e:5game} in the manner as shown below in~\eqref{eq:OneDetPlayerParityCond}; this corresponds to Equation \eqref{e:4game} in the main text when $r=0$, $s=1$. Whether there exists a \((4,2,2)\) input-conditional outcome distribution for the remaining parties \(\mathsf{A}_2,\mathsf{A}_3,\mathsf{A}_4,\mathsf{A}_5\) that satisfies the predicates in~\eqref{eq:OneDetPlayerParityCond}---thereby winning the \(5\)-cycle game with certainty---is then a linear feasibility program. (Here the numbers in \((4,2,2)\) denote the number of parties, inputs, and outputs, respectively.) The program checks to see if there exists a distribution \(P(A^{(2)}_u,A^{(3)}_v,A^{(4)}_w,A^{(5)}_y)\), for \(u,v,w,y\in\{X,Z\}\), satisfying the \((4,2,2)\) no-signaling conditions in addition to the predicates below (which translate to affine conditions on the optimization variable \(P\)). 
\begin{align}\label{eq:OneDetPlayerParityCond}
&A^{(2)}_Z\oplus A^{(3)}_X = A^{(4)}_X\oplus A^{(5)}_Z = 1\oplus r,\nonumber\\
&A^{(2)}_X\oplus A^{(3)}_Z\oplus A^{(4)}_X = A^{(3)}_X\oplus A^{(4)}_Z\oplus A^{(5)}_X = 1,\\
&A^{(2)}_X\oplus A^{(5)}_X = 1 \oplus A^{(2)}_Z\oplus A^{(3)}_Z\oplus A^{(4)}_Z\oplus A^{(5)}_Z = s.\nonumber
\end{align}
When implementing the program we encoded the inputs as \(X\rightarrow0\) and \(Z\rightarrow1\). Notice that in the implementation one needs to also include the non-negativity and normalization conditions on \(P\). The following proposition, whose proof is witnessed by an affirmative outcome of our linear feasibility program, formalizes the fact that a deterministic player is compatible with perfect play.
\begin{prop}\label{prop:OneDetPlayer_PerfectPlay}
For every \(r,s\in\{0,1\}\), there exists a \((4,2,2)\) behavior \(P_{r,s}\) satisfying the conditions in~\eqref{eq:OneDetPlayerParityCond}. Hence, the \(5\)-cycle game can be won with certainty when \(\mathsf{A}_1\) answers according to \(A^{(1)}_X=r, A^{(1)}_Z=1\oplus s\) and the other players share one four-party no-signaling resource with the distribution \(P_{r,s}\).
\end{prop}
\begin{proof}
Table~\ref{tab:Two4_wayNSboxes} defines two no-signaling boxes \(Q_0,Q_1\). Set
\begin{equation}
P_{r,s}(A^{(2)}_u,A^{(3)}_v,A^{(4)}_w,A^{(5)}_y) = Q_s(A^{(2)}_u\oplus ru,A^{(3)}_v,A^{(4)}_w,A^{(5)}_y\oplus ry).
\end{equation}
Thus the \(r=1\) branch flips \(\mathsf{A}_2\)'s and \(\mathsf{A}_5\)'s outputs precisely when their input is \(Z\). The no-signaling constraints are preserved by (local) relabelling of output. The affine support equations in Table~\ref{tab:Two4_wayNSboxes} verify~\eqref{eq:OneDetPlayerParityCond} for \(r=0\); the two \(Z\)-conditioned flips change the first line exactly as required for \(r=1\), preserve both three-player parities, and cancel in the four-player \(Z\)-parity. Substitution of \(\mathsf{A}_1\)'s output strategy then verifies all six predicates.
\end{proof}

\begin{table}[htb!]
\centering
\caption{\label{tab:Two4_wayNSboxes}The two base boxes \(Q_s(A^{(2)}_u,A^{(3)}_v,A^{(4)}_w,A^{(5)}_y)\) for \(s=0\) (top) and \(s=1\) (bottom). We have used the labelling \(X\to 0\) and \(Z\to 1\) for the questions, and \(\alpha=1/8\). These were generated by a computer executing a linear program, with an independent check that the~\eqref{eq:OneDetPlayerParityCond} predicates, no-signaling conditions, and normalisation are satisfied.}
\begin{tabular}{rrcccccccccccccccc}
\toprule
\multicolumn{1}{r}{} & \multicolumn{1}{r}{} & \multicolumn{16}{c}{Outcome Combinations}\\
\cmidrule{2-18}
\multicolumn{1}{r}{} & \multicolumn{1}{r}{} & \multicolumn{1}{c}{$0000$} & \multicolumn{1}{c}{$0001$} & \multicolumn{1}{c}{$0010$} & \multicolumn{1}{c}{$0011$} & \multicolumn{1}{c}{$0100$} & \multicolumn{1}{c}{$0101$} & \multicolumn{1}{c}{$0110$} & \multicolumn{1}{c}{$0111$} & \multicolumn{1}{c}{$1000$} & \multicolumn{1}{c}{$1001$} & \multicolumn{1}{c}{$1010$} & \multicolumn{1}{c}{$1011$} & \multicolumn{1}{c}{$1100$} & \multicolumn{1}{c}{$1101$} & \multicolumn{1}{c}{$1110$} & \multicolumn{1}{c}{$1111$} \\
\cmidrule{2-18}
\multirow{16}{*}{\rotatebox[origin=c]{90}{Question Combinations}} & $0000$ & $\alpha$ & 0 & $\alpha$ & 0 & $\alpha$ & 0 & $\alpha$ & 0 & 0 & $\alpha$ & 0 & $\alpha$ & 0 & $\alpha$ & 0 & $\alpha$ \\
 & $0001$ & 0 & $\alpha$ & $\alpha$ & 0 & 0 & $\alpha$ & $\alpha$ & 0 & 0 & $\alpha$ & $\alpha$ & 0 & 0 & $\alpha$ & $\alpha$ & 0 \\
 & $0010$ & 0 & 0 & $2\alpha$ & 0 & $2\alpha$ & 0 & 0 & 0 & 0 & $2\alpha$ & 0 & 0 & 0 & 0 & 0 & $2\alpha$ \\
 & $0011$ & 0 & 0 & $\alpha$ & $\alpha$ & $\alpha$ & $\alpha$ & 0 & 0 & $\alpha$ & $\alpha$ & 0 & 0 & 0 & 0 & $\alpha$ & $\alpha$ \\
 & $0100$ & 0 & 0 & $2\alpha$ & 0 & $2\alpha$ & 0 & 0 & 0 & 0 & $2\alpha$ & 0 & 0 & 0 & 0 & 0 & $2\alpha$ \\
 & $0101$ & 0 & 0 & $2\alpha$ & 0 & 0 & $2\alpha$ & 0 & 0 & 0 & $2\alpha$ & 0 & 0 & 0 & 0 & $2\alpha$ & 0 \\
 & $0110$ & $\alpha$ & 0 & $\alpha$ & 0 & $\alpha$ & 0 & $\alpha$ & 0 & 0 & $\alpha$ & 0 & $\alpha$ & 0 & $\alpha$ & 0 & $\alpha$ \\
 & $0111$ & $\alpha$ & 0 & $\alpha$ & 0 & 0 & $\alpha$ & 0 & $\alpha$ & 0 & $\alpha$ & 0 & $\alpha$ & $\alpha$ & 0 & $\alpha$ & 0 \\
 & $1000$ & 0 & 0 & 0 & 0 & $\alpha$ & $\alpha$ & $\alpha$ & $\alpha$ & $\alpha$ & $\alpha$ & $\alpha$ & $\alpha$ & 0 & 0 & 0 & 0 \\
 & $1001$ & 0 & 0 & 0 & 0 & 0 & $2\alpha$ & $2\alpha$ & 0 & 0 & $2\alpha$ & $2\alpha$ & 0 & 0 & 0 & 0 & 0 \\
 & $1010$ & 0 & 0 & 0 & 0 & $2\alpha$ & 0 & 0 & $2\alpha$ & 0 & $2\alpha$ & $2\alpha$ & 0 & 0 & 0 & 0 & 0 \\
 & $1011$ & 0 & 0 & 0 & 0 & $\alpha$ & $\alpha$ & $\alpha$ & $\alpha$ & $\alpha$ & $\alpha$ & $\alpha$ & $\alpha$ & 0 & 0 & 0 & 0 \\
 & $1100$ & 0 & $\alpha$ & $\alpha$ & 0 & $\alpha$ & 0 & 0 & $\alpha$ & 0 & $\alpha$ & $\alpha$ & 0 & $\alpha$ & 0 & 0 & $\alpha$ \\
 & $1101$ & 0 & $\alpha$ & $\alpha$ & 0 & 0 & $\alpha$ & $\alpha$ & 0 & 0 & $\alpha$ & $\alpha$ & 0 & 0 & $\alpha$ & $\alpha$ & 0 \\
 & $1110$ & $\alpha$ & 0 & 0 & $\alpha$ & $\alpha$ & 0 & 0 & $\alpha$ & 0 & $\alpha$ & $\alpha$ & 0 & 0 & $\alpha$ & $\alpha$ & 0 \\
 & $1111$ & 0 & $\alpha$ & $\alpha$ & 0 & $\alpha$ & 0 & 0 & $\alpha$ & $\alpha$ & 0 & 0 & $\alpha$ & 0 & $\alpha$ & $\alpha$ & 0 \\
\bottomrule
\end{tabular}
\vspace{1em}
\begin{tabular}{rrcccccccccccccccc}
\toprule
\multicolumn{1}{r}{} & \multicolumn{1}{r}{} & \multicolumn{16}{c}{Outcome Combinations}\\
\cmidrule{2-18}
\multicolumn{1}{r}{} & \multicolumn{1}{r}{} & \multicolumn{1}{c}{$0000$} & \multicolumn{1}{c}{$0001$} & \multicolumn{1}{c}{$0010$} & \multicolumn{1}{c}{$0011$} & \multicolumn{1}{c}{$0100$} & \multicolumn{1}{c}{$0101$} & \multicolumn{1}{c}{$0110$} & \multicolumn{1}{c}{$0111$} & \multicolumn{1}{c}{$1000$} & \multicolumn{1}{c}{$1001$} & \multicolumn{1}{c}{$1010$} & \multicolumn{1}{c}{$1011$} & \multicolumn{1}{c}{$1100$} & \multicolumn{1}{c}{$1101$} & \multicolumn{1}{c}{$1110$} & \multicolumn{1}{c}{$1111$} \\
\cmidrule{2-18}
\multirow{16}{*}{\rotatebox[origin=c]{90}{Question Combinations}} & $0000$ & 0 & $\alpha$ & 0 & $\alpha$ & 0 & $\alpha$ & 0 & $\alpha$ & $\alpha$ & 0 & $\alpha$ & 0 & $\alpha$ & 0 & $\alpha$ & 0 \\
 & $0001$ & 0 & $\alpha$ & $\alpha$ & 0 & 0 & $\alpha$ & $\alpha$ & 0 & 0 & $\alpha$ & $\alpha$ & 0 & 0 & $\alpha$ & $\alpha$ & 0 \\
 & $0010$ & 0 & $2\alpha$ & 0 & 0 & 0 & 0 & 0 & $2\alpha$ & 0 & 0 & $2\alpha$ & 0 & $2\alpha$ & 0 & 0 & 0 \\
 & $0011$ & $\alpha$ & $\alpha$ & 0 & 0 & 0 & 0 & $\alpha$ & $\alpha$ & 0 & 0 & $\alpha$ & $\alpha$ & $\alpha$ & $\alpha$ & 0 & 0 \\
 & $0100$ & 0 & 0 & 0 & $2\alpha$ & 0 & $2\alpha$ & 0 & 0 & $2\alpha$ & 0 & 0 & 0 & 0 & 0 & $2\alpha$ & 0 \\
 & $0101$ & 0 & 0 & $2\alpha$ & 0 & 0 & $2\alpha$ & 0 & 0 & 0 & $2\alpha$ & 0 & 0 & 0 & 0 & $2\alpha$ & 0 \\
 & $0110$ & 0 & $\alpha$ & 0 & $\alpha$ & 0 & $\alpha$ & 0 & $\alpha$ & $\alpha$ & 0 & $\alpha$ & 0 & $\alpha$ & 0 & $\alpha$ & 0 \\
 & $0111$ & $\alpha$ & 0 & $\alpha$ & 0 & 0 & $\alpha$ & 0 & $\alpha$ & 0 & $\alpha$ & 0 & $\alpha$ & $\alpha$ & 0 & $\alpha$ & 0 \\
 & $1000$ & 0 & 0 & 0 & 0 & $\alpha$ & $\alpha$ & $\alpha$ & $\alpha$ & $\alpha$ & $\alpha$ & $\alpha$ & $\alpha$ & 0 & 0 & 0 & 0 \\
 & $1001$ & 0 & 0 & 0 & 0 & 0 & $2\alpha$ & $2\alpha$ & 0 & 0 & $2\alpha$ & $2\alpha$ & 0 & 0 & 0 & 0 & 0 \\
 & $1010$ & 0 & 0 & 0 & 0 & $2\alpha$ & 0 & 0 & $2\alpha$ & 0 & $2\alpha$ & $2\alpha$ & 0 & 0 & 0 & 0 & 0 \\
 & $1011$ & 0 & 0 & 0 & 0 & $\alpha$ & $\alpha$ & $\alpha$ & $\alpha$ & $\alpha$ & $\alpha$ & $\alpha$ & $\alpha$ & 0 & 0 & 0 & 0 \\
 & $1100$ & $\alpha$ & 0 & 0 & $\alpha$ & 0 & $\alpha$ & $\alpha$ & 0 & $\alpha$ & 0 & 0 & $\alpha$ & 0 & $\alpha$ & $\alpha$ & 0 \\
 & $1101$ & 0 & $\alpha$ & $\alpha$ & 0 & 0 & $\alpha$ & $\alpha$ & 0 & 0 & $\alpha$ & $\alpha$ & 0 & 0 & $\alpha$ & $\alpha$ & 0 \\
 & $1110$ & $\alpha$ & 0 & 0 & $\alpha$ & $\alpha$ & 0 & 0 & $\alpha$ & 0 & $\alpha$ & $\alpha$ & 0 & 0 & $\alpha$ & $\alpha$ & 0 \\
 & $1111$ & $\alpha$ & 0 & 0 & $\alpha$ & 0 & $\alpha$ & $\alpha$ & 0 & 0 & $\alpha$ & $\alpha$ & 0 & $\alpha$ & 0 & 0 & $\alpha$ \\
\bottomrule
\end{tabular}
\end{table}
\subsection{Reproducing the Five-Player Quantum Behavior}
Perfect play constrains only the six predicates in~\eqref{e:4game}, so Proposition~\ref{prop:OneDetPlayer_PerfectPlay} does not imply that the four-player no-signaling boxes \(P_{r,s}\) reproduce the full five-player quantum behavior induced by the quantum state and measurements described in Section~\ref{s:qstrat}. The following decomposition shows that a convex mixture of the same boxes can reproduce every input-conditional output probability of the five-player quantum behavior.
\begin{prop}\label{prop:FourBoxRepQuantum}
The five-player quantum-permissible behavior winning the game with certainty is simulable using two bits of shared randomness and four four-way no-signaling boxes \(P_{r,s}\).
\end{prop}
\begin{proof}
Suppose the five players share two uniform bits \((r,s)\). Relabel the inputs \(X,Z\) as \(0,1\), respectively. Player \(\mathsf{A}_1\) outputs deterministically as
\begin{equation}\label{eq:A_outcome_rule}
A^{(1)}_t = (1-t)r\oplus t(1\oplus s),\,\text{for input }t\in\{0,1\},
\end{equation}
and players \(\mathsf{A}_2,\mathsf{A}_3,\mathsf{A}_4,\mathsf{A}_5\) use \(P_{s}^{(r)}\), then the resulting behavior is exactly \(P_{\mathrm Q}\):
\begin{equation}
P_{\mathrm Q}(A^{(1)}_t,A^{(2)}_u,A^{(3)}_v,A^{(4)}_w,A^{(5)}_y) = \frac{1}{4}\sum_{r,s=0}^{1}[\![A^{(1)}_t = (1-t)r\oplus t(1\oplus s)]\!]P_{r,s}(A^{(2)}_u,A^{(3)}_v,A^{(4)}_w,A^{(5)}_y),
\end{equation}
where the Iverson bracket \([\![\Phi]\!]\) serves as an indicator function, evaluating to \(1\) when the condition \(\Phi\) holds and \(0\) otherwise.
\end{proof}
\begin{remark}
Player \(\mathsf{A}_1\) is deterministic only after conditioning on \((r,s)\). Averaging the four branches
restores the non-deterministic marginals of \(P_{\mathrm Q}\).
\end{remark}
Following the so-called \textit{All-versus-Nothing} arguments, using techniques like the \textit{table of zeroes} (see, for instance, Definition.~\Rmnum{2}.~7 in~\cite{Mansfield17} and Appendix A in~\cite{LiuFNS_FN_AVN_PT}), we can verify that the distributions \(Q_0, Q_1\) are strongly nonlocal. Let \(\lambda\) be the register listing the assignments of outcomes for all possible inputs; there are \((2^2)^4=256\) such assignments.
\begin{equation}
\lambda = \begin{Bmatrix}
A^{(2)}_X = 0, A^{(3)}_X = 0, A^{(4)}_X = 0, A^{(5)}_X = 0; A^{(2)}_Z = 0, A^{(3)}_Z = 0, A^{(4)}_Z = 0, A^{(5)}_Z = 0\\
A^{(2)}_X = 0, A^{(3)}_X = 0, A^{(4)}_X = 0, A^{(5)}_X = 0; A^{(2)}_Z = 0, A^{(3)}_Z = 0, A^{(4)}_Z = 0, A^{(5)}_Z = 1\\
\vdots\\
A^{(2)}_X = 1, A^{(3)}_X = 0, A^{(4)}_X = 0, A^{(5)}_X = 1; A^{(2)}_Z = 0, A^{(3)}_Z = 1, A^{(4)}_Z = 1, A^{(5)}_Z = 1\\
\vdots\\
A^{(2)}_X = 1, A^{(3)}_X = 1, A^{(4)}_X = 1, A^{(5)}_X = 1; A^{(2)}_Z = 1, A^{(3)}_Z = 1, A^{(4)}_Z = 1, A^{(5)}_Z = 0\\
A^{(2)}_X = 1, A^{(3)}_X = 1, A^{(4)}_X = 1, A^{(5)}_X = 1; A^{(2)}_Z = 1, A^{(3)}_Z = 1, A^{(4)}_Z = 1, A^{(5)}_Z = 1\\
\end{Bmatrix}.
\end{equation}
We denote by \(A^{(i)}_s(\lambda)\) the output for party \(\mathsf{A}_i\) corresponding to input \(s\) according to the register \(\lambda\). We wrote a small Python script to verify that the distribution \(Q_s(A^{(2)}_u,A^{(3)}_v,A^{(4)}_w,A^{(5)}_y)\) satisfies the following condition for a strongly nonlocal behavior:
\begin{equation}\label{e:snq}
Q_s\text{ is strongly nonlocal }\Leftrightarrow \text{For all }\lambda,\text{ there exists }u,v,w,y,\text{ s.t., }Q_s(A^{(2)}_u(\lambda),A^{(3)}_v(\lambda),A^{(4)}_w(\lambda),A^{(5)}_y(\lambda))=0,
\end{equation}
meaning that every local deterministic strategy generates an outcome with positive probability that the behavior \(Q_{s}\) assigns zero probability. Equivalently, the set defined as:
\begin{equation}
K(Q_s)\coloneqq\{\lambda\colon Q_s(A^{(2)}_u(\lambda),A^{(3)}_v(\lambda),A^{(4)}_w(\lambda),A^{(5)}_y(\lambda))>0,\,\forall u,v,w,y\}    
\end{equation}
is empty. Since every local behavior is a convex combination of local deterministic behaviors this condition will hold for such a convex combination.

\section{Testing Membership in the Quantum Set}\label{s:qtest}
\subsection{Membership via Outer Approximation}

We want to test if the behavior \(Q_0(A^{(2)}_u,A^{(3)}_v,A^{(4)}_w,A^{(5)}_y)\)---one of the four behaviors that simulate the full five-player quantum behavior that wins the game perfectly---is itself quantum-admissible. One way to perform this test is by means of a Bell inequality violation. We check to see if there exists an outer-approximation of the set of quantum-achievable distributions such that maximizing a suitable Bell functional over the approximated set one achieves the (algebraic) maximal value.
The Bell inequality that we use in our test is the one that is algebraically saturated by \(Q_0\).

We now obtain the four-player Bell functional, used to test the membership of \(Q_0\) in the quantum set, from a five-player Bell functional by making one player respond deterministically on both inputs. First, we reformulate the five-player Bell functional \(\mathcal S=\mathbb{P}(W_f)+\sum_{i=1}^5\mathbb{P}(W_i)\), expressed in terms of probabilities of outcome-parity conditions, in terms of correlators. For any five-player behavior \(P(A^{(1)}_t,A^{(2)}_u,A^{(3)}_v,A^{(4)}_w,A^{(5)}_y)\), for \(t,u,v,w,y\in\{X,Z\}\), we can define the correlators \(\langle \pA^{(i-1)}_X\pA^{(i)}_Z\pA^{(i+1)}_X\rangle\), with \(\pA_s^{(i)}\coloneqq(-1)^{A_{s}^{(i)}}\), corresponding to the three-player parity conditions, as the difference between the even parity and odd parity of the outcomes of the three players \(\mathsf{A}_{i-1},\mathsf{A}_i,\mathsf{A}_{i+1}\):
\begin{equation}\label{eq:Correlator_ThreePlayers}
\begin{split}
\langle\pA^{(i-1)}_X\pA^{(i)}_Z\pA^{(i+1)}_X\rangle &\coloneqq \mathbb{P}_P(A^{(i-1)}_X\oplus A^{(i)}_Z\oplus A^{(i+1)}_X=0)-\mathbb{P}_P(A^{(i-1)}_X\oplus A^{(i)}_Z\oplus A^{(i+1)}_X=1)\\
&=1-2\mathbb{P}_P(A^{(i-1)}_X\oplus A^{(i)}_Z\oplus A^{(i+1)}_X=1).
\end{split}
\end{equation}
Similarly, corresponding to the the sixth predicate which is a parity condition on the \(Z\)-outcomes of all five players, the full correlator \(\langle \pA^{(1)}_Z\pA^{(2)}_Z\cdots \pA^{(5)}_Z\rangle\) is defined as follows:
\begin{equation}
\langle \pA^{(1)}_Z\pA^{(2)}_Z\cdots \pA^{(5)}_Z\rangle\coloneqq 2\mathbb{P}_P(A^{(1)}_Z\oplus\cdots\oplus A^{(5)}_Z=0)-1.
\end{equation}
The Bell functional \(\mathcal S\) can then be expressed in terms of correlators as:
\begin{equation}\label{eq:ProbBellFunc_to_CorrBellFunc}
\mathcal{S} = \frac{1}{2}\sum_{i=1}^{5}\big(1-\langle \pA^{(i-1)}_X\pA^{(i)}_Z\pA^{(i+1)}_X\rangle\big) + \frac{1}{2} + \frac{1}{2}\langle\pA^{(1)}_Z\cdots \pA^{(5)}_Z\rangle = 3 + \frac{\mathcal C}{2},
\end{equation}
where \(\mathcal{C} = -\sum_{i=1}^{5}\langle\pA^{(i-1)}_X\pA^{(i)}_Z\pA^{(i+1)}_X\rangle+\langle \pA^{(1)}_Z\pA^{(2)}_Z\cdots \pA^{(5)}_Z\rangle\). Assuming equal probability for the six winning predicates, the winning probability is \(\omega=(6+\mathcal{C})/12\). The local maximum and algebraic maximum of \(\mathcal S\) is \(5\) and \(6\), respectively, whereas the local maximum and algebraic maximum of \(\mathcal C\) is \(4\) and \(6\), respectively. When one player, say \(\mathsf{A}_1\), behaves deterministically, for instance, by declaring outputs \(A^{(1)}_X=A^{(1)}_Z=1\), then the four-player Bell functional is obtained by substituting \(\pA^{(1)}_X=\pA^{(1)}_Z=1\) in the five-player Bell functional \(\mathcal C\), which is shown below with the relabelled questions (inputs) \(X\rightarrow 0, Z\rightarrow 1\):
\begin{equation}\label{eq:4PlayerBellIneq}
\begin{split}
\mathcal{C}'\equiv-\langle \pA^{(2)}_1\pA^{(3)}_0\rangle - \langle \pA^{(2)}_0\pA^{(3)}_1\pA^{(4)}_0\rangle - \langle \pA^{(3)}_0\pA^{(4)}_1\pA^{(5)}_0\rangle &- \langle \pA^{(4)}_0\pA^{(5)}_1\rangle - \langle \pA^{(2)}_0\pA^{(5)}_0\rangle\\
&+ \langle \pA^{(2)}_1\pA^{(3)}_1\pA^{(4)}_1\pA^{(5)}_1\rangle \overset{\mathcal L}{\le} 4.
\end{split}
\end{equation}
Denote the outcome-\(0\) projectors for \(\mathsf{A}_2,\mathsf{A}_3,\mathsf{A}_4,\mathsf{A}_5\) by \(\{F_0, F_1\}, \{G_0, G_1\}, \{R_0, R_1\}, \{S_0, S_1\}\), respectively. Let \(\mathcal{A}\coloneqq\{I,F_0,F_1,G_0,G_1,R_0,R_1,S_0,S_1\}\) 
be the alphabet of operators belonging to the four players. Denote by \(\mathcal{W}\) the set of words formed by the elements of \(\mathcal A\) of length at most \(k\).  In constructing the set \(\mathcal W\), we apply the conditions satisfied by the projection operators, namely, idempotence, $XX=X$, and commutation, \(XY=YX\), for \(X,Y\in\mathcal{A}\) not belonging to the same party. (We do not assume within-party operators to necessarily commute, for instance, \(F_0F_1\neq F_1F_0\).) \(\Gamma\) is the \(|\mathcal{W}|\times|\mathcal{W}|\) (moment) matrix with entries \(\Gamma_{u,v}=\mathrm{Tr}[u^{\dagger}v\rho]\), for \(u,v\in\mathcal{W}\). We can write the Bell expression \eqref{eq:4PlayerBellIneq} in terms of the elements of the moment matrix $\Gamma$, as follows. While the set $\mathcal A$ contains only outcome-0 projectors, the outcome-1 projectors can be obtained from these as orthogonal complements; i.e., subtraction from $I$ which is also in $\mathcal{A}$. The Bell expression, when written as a linear combination of probability terms, has terms like \(p(0,1,0,1|u,v,w,y)\) which can be expressed as \(\mathrm{Tr}[F_u(I-G_v)R_w(I-S_y)\rho]\). All probability terms, be it a full probability term or a marginal probability term, can be expressed like this. The Bell-functional \(\mathcal{C}'\) in ~\eqref{eq:4PlayerBellIneq} can then be re-expressed in terms of the moment matrix entries \(\Gamma_{i,j}\):
\begin{equation}\label{eq:4PlayerBellIneq_MomentMat}
\begin{split}
\mathcal{C}'(\Gamma) &= -4\Gamma_{I,G_1}-4\Gamma_{I,R_1}-4\Gamma_{F_1,G_0}+4\Gamma_{F_0,G_1}+4\Gamma_{F_1,G_1}+4\Gamma_{F_0,R_0}+4\Gamma_{G_1,R_0}+4\Gamma_{F_1,R_1}+4\Gamma_{G_0,R_1}\\ 
&+ 4\Gamma_{G_1,R_1} - 4\Gamma_{F_0,S_0} + 4\Gamma_{G_0,S_0} + 4\Gamma_{R_1,S_0} + 4\Gamma_{F_1,S_1} + 4\Gamma_{G_1,S_1} - 4\Gamma_{R_0,S_1} + 4\Gamma_{R_1,S_1} - 8\Gamma_{F_0,G_1R_0}\\ 
&- 8\Gamma_{F_1,G_1R_1} - 8\Gamma_{G_0,R_1S_0}-8\Gamma_{F_1,G_1S_1}-8\Gamma_{F_1,R_1S_1}-8\Gamma_{G_1,R_1S_1}+16\Gamma_{F_1G_1,R_1S_1}.
\end{split}
\end{equation}
Notice that the expression in~\eqref{eq:4PlayerBellIneq_MomentMat} is not unique due to the idempotence and inter-party commutation relations. For instance, \(\Gamma_{G_1,R_1S_1}=\Gamma_{G_1R_1,S_1}=\Gamma_{G_1R_1,R_1S_1}\). When implementing the NPA hierarchy, one needs to include the several equality conditions of the form \(\Gamma_{u,v}=\Gamma_{s,t}\), whenever \(u^\dagger v\sim s^\dagger t\), for \(u,v,s,t\in\mathcal{W}\). The expression in~\eqref{eq:4PlayerBellIneq_MomentMat} was written considering the second level (i.e., \(k=2\)) of the NPA hierarchy. It is valid for any higher level since each moment matrix entry in the expression is equal to a valid entry of the moment matrix that reduces to it. For instance, \(\Gamma_{F_1G_1,R_1S_1}=\Gamma_{F_1,G_1R_1S_1}\), where terms like \(\Gamma_{F_1,G_1R_1S_1}\) come into play only when we implement the third level of the hierarchy. 

The NPA procedure is to optimize \eqref{eq:4PlayerBellIneq_MomentMat}, subject to the constraint that $\Gamma$ is positive semidefinite with its entries satisfying the various identities described above. After executing the NPA heirarchy at the second level and finding an upper bound value of approximately 5.03, we executed the NPA hierarchy at the third level and found a numerical upper bound of \(4.828427127124\) which shows agreement up to eight places after decimal with \(2(1+\sqrt{2})\approx 4.828427124746\). We also increased the level-3 word list by adding the list \(\{F_1G_1R_1S_1, F_1G_1R_1S_0, F_1G_1R_0S_0, F_1G_1R_0S_1, F_1G_0R_0S_0, F_1G_0R_0S_1, F_1G_0R_1S_0, F_1G_0R_1S_1\}\) of operator words of length four, thereby obtaining a slightly tighter outer-approximation of the quantum set, for which we achieved the value \(4.828427125302\). While the NPA bounds apply only to an outer approximation of the quantum set, the value of~\eqref{eq:4PlayerBellIneq} gradually approaching \(2(1+\sqrt{2})\) indicates that the quantum maximal bound of~\eqref{eq:4PlayerBellIneq} is \(2(1+\sqrt{2})\). Translating from the value of the correlator to the value of the winning probability using the relationship $\omega=(6+\mathcal C)/12$ described earlier, we obtain the winning probability bound $(4+\sqrt{2})/6$ presented in the main text. In the next section, we show how a direct numerical search of explicit quantum strategies is used to confirm that quantum mechanics does in fact saturate the aforementioned bound.

\subsection{Membership via Numerical Procedure}
If the value of the Bell functional found from the NPA hierarchy is found equal to an explicit achievable quantum value---which could be the result of some calculation, then we can claim to have found the quantum maximal value. Here we describe a numerical approach to find the quantum maximal value of \(\mathcal{C}'\).
Since in our scenario all four players receive two questions each of which has two possible outcomes, their local measurements correspond to pairs of dichotomic observables. Invoking Jordan's lemma allows us to restrict the shared state to \((\mathbb{C}^{2})^{\otimes 4}\) without loss of generality, and by a suitable choice of local basis, we can parametrize the each player's observables strictly in terms of the Pauli \(\sigma_{\mathsf X}\) and \(\sigma_{\mathsf Z}\) matrices. Without loss generality, for player \(\mathsf{A}_i\), \(i\in\{2,\ldots,5\}\), we can consider the observable on input \(X\) (relabelled as `\(0\)') to be \(O^{(i)}_0=Z\) and the observable on input \(Z\) (relabelled as `\(1\)') to be \(O^{(i)}_1=(\cos\theta_i)\sigma_{\mathsf Z}+(\sin\theta_i)\sigma_{\mathsf X}\). Furthermore, we need not parametrize the state \(\rho\). We justify these two assumptions as follows:
Consider a generic pair of observables \(O_0, O_1\). By the spectral theorem, we can express \(O_0=|u_{+}\rangle\langle u_{+}|-|u_{-}\rangle\langle u_{-}|\), where \(\{|u_{+}\rangle,|u_{-}\rangle\}\) is an orthonormal basis satisfying \(|u_{+}\rangle\langle u_{+}|+|u_{-}\rangle\langle u_{-}|=I\) and \(\langle u_{+}|u_{-}\rangle=0\). We can construct a local unitary transformation \(U_0=|0\rangle\langle u_{+}|+|1\rangle\langle u_{-}|\) that rotates \(O_0\) to \(\sigma_{\mathsf Z}\). Let \(O'_1=U_0O_1U^\dagger_0\). Since \(O'_1\) is Hermitian (as a unitary transformation of a Hermitian operator) and satisfies \((O'_1)^2 = I\), it can be expressed as
\begin{equation}\label{e:wcrule}
O'_{1} = \begin{bmatrix}c & w\\ w^\ast & -c\end{bmatrix},\,\,\text{with }c^2+|w|^2=1.
\end{equation}
Therefore, we can take \(c=\cos\theta\) and \(w=e^{-i\phi}\sin\theta\), for \(\theta\in[0,\pi], \phi\in[0,2\pi]\). Now let \(U_1=DU_0\), where \(D=e^{i\phi/2}|0\rangle\langle 0|+e^{-i\phi/2}|1\rangle\langle 1|\). We have \(U_1O_0U^\dagger_1=DU_0O_0U^\dagger_0D^\dagger_0=D\sigma_{\mathsf Z}D^\dagger=\sigma_{\mathsf Z}\). And, \(U_1O_1U^\dagger_1=DU_0O_1U^\dagger_0D^\dagger=DO'_1D^\dagger=(\cos\theta)\sigma_{\mathsf Z}+(\sin\theta)\sigma_{\mathsf X}\). The existence of such a local unitary transformation \(U_1\) justifies the choice \(O_0=\sigma_{\mathsf Z}\) and \(O_1=(\cos\theta)\sigma_{\mathsf Z}+(\sin\theta)\sigma_{\mathsf X}\). All four players can adopt this strategy for constructing a local transformation leading to a global unitary transformation \(V=U_1\otimes U_2\otimes U_3\otimes U_4\). The Bell functional \(\mathcal{C}'\) is expressible as a linear combination of observable terms \(O^{(2)}_u\otimes O^{(3)}_v\otimes O^{(4)}_w\otimes O^{(5)}_y\). One can apply the same transformation \(V\) to both \(\mathcal{C}'\) and \(\rho\) and observe that \(\mathrm{Tr}[V\mathcal{C}'V^\dagger V\rho V^{\dagger}]=\mathrm{Tr}[\mathcal{C}'\rho]\). Hence, arbitrary qubit measurements on a state $\rho$ can always be re-expressed equivalently as $\sigma_{\mathsf Z}$ and $(\cos\theta)\sigma_{\mathsf Z}+(\sin \theta)\sigma_{\mathsf X}$ measurements on a different state $V\rho V^\dagger$. The problem of finding the maximum quantum value for \(\mathrm{Tr}[\mathcal{C}'\rho]\) is then
\begin{equation}\label{eq:MaxQVal_optimization}
\max_{\boldsymbol{\theta}\in[0,\pi]^4}\mathrm{Tr}[\mathcal{C}'(\boldsymbol{\theta})\rho] = \max_{\boldsymbol{\theta}\in[0,\pi]^4}\lambda_{\max}(\mathcal{C}'(\boldsymbol{\theta})),
\end{equation}
where \(\boldsymbol{\theta}\equiv(\theta_1,\theta_2,\theta_3,\theta_4)\), and $\lambda_{\max}(\cdot)$ denotes the maximal eigenvalue of the enclosed operator (well-defined when the operator is Hermitian and so has real eigenvalues). The equality in~\eqref{eq:MaxQVal_optimization} follows from considering the eigen-decomposition of \(\mathcal{C}'(\boldsymbol{\theta})\). Let \(\mathcal{C}'(\boldsymbol{\theta})=\sum_{k}\lambda_k(\boldsymbol{\theta})|w_k(\boldsymbol{\theta})\rangle\langle w_k(\boldsymbol{\theta})|\), where \(\{|w_k(\boldsymbol{\theta})\rangle\}\) is an eigenbasis. Then \(\mathrm{Tr}[\sum_{k}\lambda_k(\boldsymbol{\theta})|w_k(\boldsymbol{\theta})\rangle\langle w_k(\boldsymbol{\theta})|\rho]=\sum_{k}\lambda_k(\boldsymbol{\theta})\langle w_k(\boldsymbol{\theta})|\rho|w_k(\boldsymbol{\theta})\rangle\le\lambda_{\max}(\mathcal{C}'(\boldsymbol{\theta}))\). The inequality follows from the fact that \(\langle w_k(\boldsymbol{\theta})|\rho|w_k(\boldsymbol{\theta})\rangle\ge 0\) specifies a probability distribution and the average of all eigenvalues of \(\mathcal{C}'(\boldsymbol{\theta})\) cannot exceed its maximum eigenvalue. Now an optimization searching for the maximal eigenvalue no longer involves $\rho$ (though if the optimal eigenvalue is found, the corresponding eigenvector tells us a choice for an optimal $\rho$). This justifies the earlier claim that it is not necessary to parametrize $\rho$ to pursue the optimization.

Direct numerical optimization of~\eqref{eq:MaxQVal_optimization} is not guaranteed to converge to a globally optimal solution, but if a solver finds a quantum value in excess of the local bound of four, this conclusively demonstrates the existence of an explicit quantum strategy for beating the local bound. Our optimization yielded a quantum realisation beating the local bound with a numerical value \(4.82842712474\) appearing very close to \(2+2\sqrt{2}\). The measurement angles obtained from our numerical optimization are \(\boldsymbol{\theta}_\ast\equiv(1.57079625367, 1.57079627710, 1.57079633988, 1.57079627438)\), which appear to be very close to the value \(\pi/2\). By inspection with the aid of a computer algebra program we were able to translate this into an exact strategy as described below.

Suppose \(O_s^{(i)}\) is an observable for party \(\mathsf{A}_i\) corresponding to the measurement setting label \(s\in\{0,1\}\). As explained earlier in the subsection, we can consider \(O^{(i)}_0=\sigma_{\mathsf Z}\) and \(O^{(i)}_{1}(\theta_i)=(\cos\theta_i)\sigma_{\mathsf Z}+(\sin\theta_i)\sigma_{\mathsf X}\). Then from the numerical optimization, from which we inferred that \(\boldsymbol{\theta}_\ast\equiv(\pi/2,\pi/2,\pi/2,\pi/2)\), we see that the optimal measurement strategy for party \(\mathsf{A}_i\), for each \(i\), is to measure \(O^{(i)}_0=\sigma_{\mathsf Z}\) and \(O^{(i)}_1=\sigma_{\mathsf X}\). To find the optimal (pure) state, we need to inspect the eigenvector corresponding to the largest eigenvalue of the Bell operator corresponding to the Bell functional \(\mathcal{C}'\) in~\eqref{eq:4PlayerBellIneq}, which is shown below.
\begin{align}
\mathcal{O}(\boldsymbol{\theta}) =& -O^{(2)}_1(\theta_2)\otimes O^{(3)}_0\otimes I\otimes I - O^{(2)}_0\otimes O^{(3)}_1(\theta_3)\otimes O^{(4)}_0\otimes I - I\otimes O^{(3)}_0\otimes O^{(4)}_1(\theta_4)\otimes O^{(5)}_0\notag\\
&- I\otimes I\otimes O^{(4)}_0\otimes O^{(5)}_1(\theta_5)
- O^{(2)}_0\otimes I\otimes I\otimes O^{(5)}_0 + O^{(2)}_1(\theta_2)\otimes O^{(3)}_1(\theta_3)\otimes O^{(4)}_1(\theta_4)\otimes O^{(5)}_1(\theta_5),\notag\\
\Rightarrow \mathcal{O}(\boldsymbol{\theta}_\ast) =& - \sigma_{\mathsf X}\otimes\sigma_{\mathsf Z}\otimes I\otimes I - \sigma_{\mathsf Z}\otimes\sigma_{\mathsf X}\otimes\sigma_{\mathsf Z}\otimes I - I\otimes\sigma_{\mathsf Z}\otimes\sigma_{\mathsf X}\otimes\sigma_{\mathsf Z} - I\otimes I\otimes\sigma_{\mathsf Z}\otimes\sigma_{\mathsf X}-\sigma_{\mathsf Z}\otimes I\otimes I\otimes\sigma_{\mathsf Z}\notag\\
&+\sigma_{\mathsf X}\otimes\sigma_{\mathsf X}\otimes\sigma_{\mathsf X}\otimes\sigma_{\mathsf X}.
\end{align}
The six eigenvalues of \(\mathcal{O}(\pi/2,\pi/2,\pi/2,\pi/2)\) are \(\{2\pm 2\sqrt{2},-2\pm2\sqrt{2},\pm 2\}\). The pure state that is part of the optimal strategy is the (normalised) eigenvector corresponding to the largest eigenvalue \(2+2\sqrt{2}\) as shown below:
\begin{equation}\label{eq:OptimalPureState}
|\psi_\ast\rangle = \begin{bmatrix}-b & a & b & a & b & -a & b & a & a & -b & -a & -b & a & -b & a & b\end{bmatrix}^T,
\end{equation}
where \(a=\big(1+\sqrt{2}\big)/\big(4\sqrt{2+\sqrt{2}}\big)\) and \(b=1/\big(4\sqrt{2+\sqrt{2}}\big)\). The state in~\eqref{eq:OptimalPureState} can be re-expressed as:
\begin{equation}\label{e:4state}
\begin{split}
|\psi_\ast\rangle &= \frac{1}{2}\Big(|0\!-\!0\alpha\rangle + |0\text{+}1\beta\rangle + |1\text{+}0\gamma\rangle - |1\!-\!1\delta\rangle\Big),\\
\text{where }|\alpha\rangle &= R_y(\pi/4)|1\rangle = \begin{bmatrix}
    -\sin(\pi/8) & \cos(\pi/8)
\end{bmatrix}^T, \,|\beta\rangle = R_y(-\pi/4)|1\rangle = R_y(\pi/4)|+\rangle = \begin{bmatrix}
    \sin(\pi/8) & \cos(\pi/8)
\end{bmatrix}^T,\\
|\gamma\rangle &= R_y(-\pi/4)|0\rangle = R_y(\pi/4)|-\rangle  = \begin{bmatrix}
    \cos(\pi/8) & -\sin(\pi/8)
\end{bmatrix}^T,\,
|\delta\rangle = R_y(\pi/4)|0\rangle = \begin{bmatrix}
    \cos(\pi/8) & \sin(\pi/8)
\end{bmatrix}^T.
\end{split}
\end{equation}
Thus, the optimal quantum strategy achieving the violation of \(2+2\sqrt{2}\) is \(\{|\psi_\ast\rangle,\{O^{(i)}_0=\sigma_{\mathsf{Z}},O^{(i)}_1=\sigma_{\mathsf{X}}\}\}\). As described in the previous subsection, this violation value translates to the winning probability of $(4+\sqrt{2})/6$ cited in the main text, in excess of the local bound of $5/6$.

We briefly remark that the state in~\eqref{e:4state} yields the given Bell violation in a somewhat interesting manner. Two of the correlators yield a value of $-1$, while the remaining four yield $\pm1/\sqrt{2}$; this might seem to suggest alternation between an alignment game ($\pm 1$ correlators) and a distributed CHSH game (four $\pm 1/\sqrt{2}$ values), as can be discerned in, for example, Expression (3) of~\cite{mao:2022} (which pertains to a three-party scenario with a shared GHZ state, in which one party measures $\sigma_{\mathsf X}$ to swap a Bell state to the other two parties for them to violate CHSH, or alternatively measures $\sigma_{\mathsf Z}$ to align with another party). However any possible interpretation of the current example as containing an embedded CHSH scenario would necessarily be more complicated. To illustrate, if we apply the change of basis induced by the Hadamard matrix to the second qubit of the state in \eqref{e:4state}, it transforms as $\ket + \to \ket 0$ and $\ket - \to \ket 1$ while the corresponding measuring operators $\sigma_{\mathsf X}$ and $\sigma_{\mathsf Z}$ are interchanged. This allows re-writing of the state as
\begin{align}\label{e:newstate}
\ket{\psi_{\ast}} &= (1/2)(\ket{010\alpha} + \ket{001\beta} + \ket{100\gamma}-\ket{111\delta} )\nonumber\\
&= (1/2)(\ket{\gamma \text{++}1} + \ket{\alpha \text{+}\!-\!0} + \ket{\beta\!-\!\text{+}0}-\ket{\delta\!-\!-1}
\end{align}
where the measurements are now $\sigma_{\mathsf X}\sigma_{\mathsf X}II$, $II\sigma_{\mathsf Z}\sigma_{\mathsf X}$, $\sigma_{\mathsf Z}\sigma_{\mathsf Z}\sigma_{\mathsf Z}I$, $\sigma_{\mathsf Z}II\sigma_{\mathsf Z}$, $I\sigma_{\mathsf X}\sigma_{\mathsf X}\sigma_{\mathsf Z}$, and $\sigma_{\mathsf X}\sigma_{\mathsf Z}\sigma_{\mathsf X}\sigma_{\mathsf X}$. 
(Incidentally, all correlators can be observed with only four compatible all-party measurements comprising $\sigma_{\mathsf X}\sigma_{\mathsf X}\sigma_{\mathsf Z}\sigma_{\mathsf X}$, $\sigma_{\mathsf Z}\sigma_{\mathsf Z}\sigma_{\mathsf Z}\sigma_{\mathsf Z}$, $\sigma_{\mathsf Z}\sigma_{\mathsf X}\sigma_{\mathsf X}\sigma_{\mathsf Z}$, and $\sigma_{\mathsf X}\sigma_{\mathsf Z}\sigma_{\mathsf X}\sigma_{\mathsf X}$.) The two correlators yielding $-1$ are $\sigma_{\mathsf Z}\sigma_{\mathsf Z}\sigma_{\mathsf Z}I$ and $I\sigma_{\mathsf X}\sigma_{\mathsf X}\sigma_{\mathsf Z}$, as is made clearer with the two respective expressions in~\eqref{e:newstate}. 
And measuring either $\sigma_{\mathsf Z}\sigma_{\mathsf Z}$ or $\sigma_{\mathsf X}\sigma_{\mathsf X}$ on the middle two qubits projects the outer two qubits into a ``near" alignment for the compatible $\sigma_{\mathsf Z}II\sigma_{\mathsf Z}$ measurement to yield its value of $-1/\sqrt 2$; alternatively, $\sigma_{\mathsf X}\sigma_{\mathsf Z}$ or $\sigma_{\mathsf Z}\sigma_{\mathsf X}$ measurements of the middle two qubits yield the other three $\pm1/\sqrt{2}$ terms in the other correlators. So the toggling between alignment outcomes and $1/\sqrt 2$ outcomes thus here reflects an intricate interplay between the \textit{joint} measurements of the two central parties, and their interdependencies with the outer two parties.

\section{Impossibility of Svetlichny-type \texorpdfstring{\(2|3\)}{23} No-Signaling Models for the 5-Cycle Game}\label{s:svet}
The no-signaling Svetlichny model \cite{svetlichny:1987,Bancal13} for a five-partite behavior that is replicable with tripartite resources is as follows. For a globally shared classical random variable \(\lambda\), let \(S_\lambda\subset\{\mathsf{A,B,C,D,E}\}\) denote a \(\lambda\)-indexed collection of sets with \(|S_\lambda|=2\), such that the observed five-partite behavior has the form
\begin{equation}\label{e:Svetlichny3}
\sum_{\lambda} q(\lambda)
P_{S_\lambda}(O_I)P_{S^{\mathrm c}_\lambda}(O^{\mathrm c}_{I^{\mathrm c}}),
\end{equation}
where \(q(\lambda)\) is the probability distribution of \(\lambda\), \(P_{S_\lambda}\) (\(P_{S^{\mathrm c}_\lambda}\)) is the no-signaling behavior with \(O\) (\(O^{\mathrm c}\)) and \(I\) (\(I^{\mathrm c}\)) representing the outcome and input combination, respectively, of the set \(S_{\lambda}\) (\(S_{\lambda}^{\mathrm c}\)). A behavior is genuinely four-partite nonlocal according to the no-signaling-Svetlichny criterion~\cite{svetlichny:1987,Bancal13} if it does \textit{not} admit a \(2|3\) mixture decomposition of the form~\eqref{e:Svetlichny3}.
\begin{prop}\label{p:Svetlichny_NoBetterThan_Local}
For the \(5\)-cycle game, arbitrary bipartite and tripartite no-signaling resources arranged according to Svetlichny-type \(2|3\) resource partitions provide no advantage over local deterministic strategies.
\end{prop}
\begin{proof}
Any no-signaling 2-party behavior \(P_{S^\lambda}\) in~\eqref{e:Svetlichny3} can be expressed as convex combination of behaviors which are either local deterministic or PR boxes~\cite{barrett:2005}. We can assume, without loss of generality, that each \(P_{S^\lambda}\) is one of these, by absorbing the convex weight over extremal no-signaling behaviors into \(\lambda\). 
Now if \(P_{S_\lambda}(O_I)\) is local deterministic, \(P_{S_\lambda}(O_I)P_{S^{\mathrm c}_\lambda}(O^{\mathrm c}_{I^{\mathrm c}})\) is a distribution where two of the parties have fixed output, and so the winning probability is bounded by \(5/6\) by Proposition~\ref{p:2DetPlayers}. On the other hand, if \(P_{S_\lambda}(O_I)\) is a PR box, such a bipartite behavior has uniform output marginals, which, combined with the independence of the \(\lambda\)-indexed distribution across the \(S_\lambda|S^\mathrm{c}_\lambda\) split due to the factoring structure of the constituent 5-party distribution, makes it impossible to satisfy the predicate of~\eqref{e:5game} that contains parties in both \(S_\lambda\) and \(S^\mathrm{c}_\lambda\), at least with probability more than 5/6. The score is linear, so the result then extends to their convex mixtures and subsequently to the average over the shared randomness \(\lambda\). 

We explicitly prove the above claim about PR box case for the partitions \(S|S^{\mathrm c}=\{\mathsf{A}_1,\mathsf{A}_2\}|\{\mathsf{A}_3,\mathsf{A}_4,\mathsf{A}_5\}\) and \(\{\mathsf{A}_1,\mathsf{A}_3\}|\{\mathsf{A}_2,\mathsf{A}_4,\mathsf{A}_5\}\). This effectively covers all distinct cases via the cyclic symmetry of the predicates in \eqref{e:5game}. Let us start with the $S _\lambda=\{\mathsf{A}_1,\mathsf{A}_2\}$ case. The effective behavior of the adjacent players \(\mathsf{A}_1,\mathsf{A}_2\) is a PR box, which is defined, up to relabelling of inputs \(t,u\in\{0,1\}\) and outcomes \(a,b\in\{0,1\}\), as \(\mathrm{PR}(A^{(1)}_t=a,A^{(2)}_u=b)=[\![a\oplus b = t\cdot u]\!]/2\).) Since a PR box has uniform single-party marginals,
\begin{equation}
\mathbb{P}(A^{(1)}_t = 0) = \mathbb{P}(A^{(1)}_t = 1) = \frac{1}{2},\quad \mathbb{P}(A^{(2)}_u = 0) = \mathbb{P}(A^{(2)}_u = 1) = \frac{1}{2}.
\end{equation}
Of the five cyclic predicates, we find two where one party belongs to the \(\mathsf{A}_1\mathsf{A}_2\)-block and the remaining two parties belong to the \(\mathsf{A}_3\mathsf{A}_4\mathsf{A}_5\)-block. They are \(A^{(2)}_X\oplus A^{(3)}_Z\oplus A^{(4)}_X = 1\) and \(A^{(4)}_X\oplus A^{(5)}_Z\oplus A^{(1)}_X=1\). Since \(A^{(2)}_X\) is uniform and independent of the \(\mathsf{A}_3\mathsf{A}_4\mathsf{A}_5\)-block, and \(A^{(1)}_X\) is uniform and independent of the \(\mathsf{A}_3\mathsf{A}_4\mathsf{A}_5\)-block, we have:
\begin{equation}
\mathbb{P}(A^{(2)}_X\oplus A^{(3)}_Z\oplus A^{(4)}_X = 1) = \frac{1}{2},\quad
\mathbb{P}(A^{(4)}_X\oplus A^{(5)}_Z\oplus A^{(1)}_X = 1) = \frac{1}{2}.
\end{equation}
Each of the other four predicates has success probability at most \(1\). Therefore, the sum of the success probabilities $\mathcal S$ satisfies
\(\mathcal{S}\le 4\times 1 + 2\times\frac{1}{2} = 5\), or equivalently, the win probability $\omega$ satisfies \(\omega\le 5/6\).

Next, we prove the same result for the non-adjacent players \(\mathsf{A}_1,\mathsf{A}_3\). We observe that four of the five cyclic predicates have one player belonging to the \(\mathsf{A}_1\mathsf{A}_3\)-block and the remaining two belonging to the \(\mathsf{A}_2\mathsf{A}_4\mathsf{A}_5\)-block. They are \(A^{(2)}_X\oplus A^{(3)}_Z\oplus A^{(4)}_X=1,A^{(3)}_X\oplus A^{(4)}_Z\oplus A^{(5)}_X=1, A^{(4)}_X\oplus A^{(5)}_Z\oplus A^{(1)}_X=1, A^{(5)}_X\oplus A^{(1)}_Z\oplus A^{(2)}_X=1\). Again, the \(2\)-party distribution is a PR box, whereby \(A^{(1)}_t\) is uniform and independent of the \(\mathsf{A}_2\mathsf{A}_4\mathsf{A}_5\)-block, for \(t\in\{X,Z\}\), and \(\mathsf{A}^{(3)}_v\) is uniform and independent of the \(\mathsf{A}_2\mathsf{A}_4\mathsf{A}_5\)-block for \(v\in\{X,Z\}\). We now have:
\begin{equation}
\begin{split}
\mathbb{P}(A^{(2)}_X\oplus A^{(3)}_Z\oplus A^{(4)}_X=1)=\frac{1}{2},\quad \mathbb{P}(A^{(3)}_X\oplus A^{(4)}_Z\oplus A^{(5)}_X=1)=\frac{1}{2},\\
\mathbb{P}(A^{(4)}_X\oplus A^{(5)}_Z\oplus A^{(1)}_X=1)=\frac{1}{2},\quad \mathbb{P}(A^{(5)}_X\oplus A^{(1)}_Z\oplus A^{(2)}_X=1)=\frac{1}{2}.
\end{split}
\end{equation}
Each of the remaining two predicates have success probability at most \(1\). So \(\mathcal{S}\le 4\times\frac{1}{2}+2\times 1=4\), and \(\omega\le 2/3\).
\end{proof}

\section{Chained-Bell Approach Generating \texorpdfstring{\(\epsilon\)}{eps}-close Pseudotelepathy Games}\label{a:ChainBell}

\subsection{Chained-Bell Quantities and Monogamy Relationships}

The well-known chained Bell inequality bounds a quantity that can be expressed in the form:
\begin{equation}\label{e:chainbell}
\begin{split}
\mathsf{CB}=
&\prbP(A_0 = B_0)+\prbP(A_1 = B_0)+\prbP(A_1 = B_1)+\prbP(A_2 = B_1)+\prbP(A_2 = B_2)+\prbP(A_3 = B_2)+ \cdots\\
&\cdots + \prbP(A_{N-2} = B_{N-2})+ \prbP(A_{N-1} = B_{N-2}) + \prbP(A_{N-1} = B_{N-1})+ \prbP(A_{0} \ne B_{N-1}).
\end{split}
\end{equation}
There are $2N$ terms, with $2N-1$ of them of the form $\prbP(A_i=B_j)$ with $i=j$ or $i+1=j$ and one special term $\prbP(A_0\ne B_{N-1})$. It is well known that for a local hidden variable theory, the quantity $\mathsf{CB}$ in \eqref{e:chainbell} is bounded above by $2N-1$, and that for arbitrarily small $\epsilon>0$, there exists a sufficiently large $N$ for which there is a quantum strategy that can achieve a value of $\mathsf{CB}$ exceeding $2N - \epsilon$, where $2N$ is the algebraic maximum corresponding to all probabilities $\prbP(\cdot)$ equalling one.

We first derive a monogamy relationship for the chained Bell quantity $\mathsf{CB}$, which we will enlist in our main argument later. Observe that applying Lemma \ref{l:Lemma_for_Prop1} with the parameters $\gamma=1$ and $\delta=1$, then using the fact that $\prbP(S\oplus T=1)=\prbP(S\ne T) = 1-\prbP(S=T)$ and moving that term to the other side of the inequality, one obtains the valid expression \eqref{e:pqrprob1} 
\begin{align}
\prbP(R\ne S)&\ge \prbP(R\ne T)+P(T = S)-1\label{e:pqrprob1}\\
\prbP(R\ne S)&\ge \prbP(R= T)+P(T \ne S)-1\label{e:pqrprob2},
\end{align}
while \eqref{e:pqrprob2} follows from \eqref{e:pqrprob1} by permutation of the variables $R\leftrightarrow S$. Introducing a third party $\mathsf C$ with setting $0$, we can write
\begin{equation*}
\begin{split}
1 &= \prbP(A_0=C_0) + \prbP(A_0\ne C_0)\\
  &\ge \prbP(A_0=C_0) + [-1+\prbP(A_0=B_0) +\prbP(B_0\ne C_0)]\\
  &\ge \prbP(A_0=C_0) -1+\prbP(A_0=B_0) +[-1+\prbP(B_0=A_1)+\prbP(A_1\ne C_0)]\\
  & \;\;\vdots
\end{split}
\end{equation*}
where the vertical dots indicate repeated replacement of terms, via \eqref{e:pqrprob2}, of the alternating form $\prbP(A_i\ne C_0) \to  \prbP(A_i=B_{i}) +\prbP(B_{i}\ne C_0) -1$ and $\prbP(B_i\ne C_0) \to  \prbP(B_i=A_{i+1}) +\prbP(A_{i+1}\ne C_0) -1$. Each iteration of this process generates a new term found in the quantity $\mathsf{CB}$ in \eqref{e:chainbell}, starting with the first term and moving one term to the right for the next iteration. Let us use the symbol ``$\Xi$" to denote the $2N-1$ terms of $\mathsf{CB}$, so that $\Xi$ is everything in \eqref{e:chainbell} \textit{except} $\prbP(A_0\ne B_{N-1})$. Then the result of the iterative process after $2N-1$ steps is
\begin{equation*}
1 \ge \prbP(A_0=C_0) + (2N-1)\times(-1) + \Xi + \prbP(B_{N-1}\ne C_0).
\end{equation*}
Now we can apply \eqref{e:pqrprob1} to $\prbP(B_{N-1}\ne C_0)$ to replace it with the not-greater quantity $\prbP(B_{N-1}\ne A_0) +\prbP(A_0= C_0)-1$, thereby generating the last remaining term of $\mathsf{CB}$. Moving the $2N$ copies of ``$-1$" to the left side of the inequality yields our desired monogamy relationship, applicable to any behavior obeying the no-signaling principle:
\begin{equation}\label{e:cbmono}
    2N+1\ge \mathsf{CB} + 2\prbP(A_0=C_0).
\end{equation}
We also require a modified form of \eqref{e:chainbell} and \eqref{e:cbmono}. Starting with \eqref{e:chainbell} and relabeling each $B_j\to B_{N-1-j}$, and then for all $A_i$ variables performing a swap-relabeling of their outputs so that $\prbP(A_i=B_j)$ terms become $\prbP(A_i\ne B_j)$ terms and vice-versa, we obtain the modified quantity 
\begin{multline}\label{e:chainbellprime}
\mathsf{CB}'= \prbP(A_0 \ne B_{N-1})+\prbP(A_1 \ne B_{N-1})+\prbP(A_1 \ne B_{N-2})+\prbP(A_2 \ne B_{N-2})+ \cdots\\
\cdots + \prbP(A_{N-2} \ne B_{1})+ \prbP(A_{N-1} \ne B_{1}) + \prbP(A_{N-1} \ne B_{0})+ \prbP(A_{0} = B_{0})
\end{multline}
for which the standard chained Bell bound $\mathsf{CB'}\le 2N-1$ will still hold since we have merely re-labeled variables and outcomes. Furthermore, a similarly relabeled monogamy inequality \eqref{e:cbmono} would now end with the term $2\prbP(A_0\ne C_0)$ due to the swap-relabeling of $A$'s outputs, but if we also swap-relabel $C$'s outputs, we get the following valid monogamy relationship:
\begin{equation}\label{e:cbmonoprime}
    2N+1\ge \mathsf{CB}' + 2\prbP(A_0=C_0).
\end{equation}

\subsection{A Chained-Bell Inspired Family of Arbitrarily Near-Pseudotelepathic Games}\label{s:cbinflationproof}

Let us consider the game depicted in Figure \ref{fig:ChainBellGHZ}, in which three parties $\mathsf{A},\mathsf{B},\mathsf{C}$ measure a GHZ state $(\ket{000}+\ket{111})/\sqrt{2}$. $\mathsf{C}$ has two measurement settings: $C_0=C_Z$ in which the Pauli measurement $Z$ collapses the $\mathsf A$-$\mathsf B$ state into $\ket{00}$ or $\ket{11}$, and $C_1=C_X$ in which the Pauli measurement $X$ collapses the $\mathsf A$-$\mathsf B$ state into one of the two Bell states $\ket{\Phi^+}=(|00\rangle + |11\rangle)/\sqrt 2$ or $\ket{\Phi^-}=(|00\rangle - |11\rangle)/\sqrt 2$. When $C_Z$ is selected, all three players can align their outcomes if $\mathsf A$ and $\mathsf B$ measure $Z$; when $C_X$ is selected, Alice and Bob can use their shared Bell state to play a chained Bell game and obtain the optimal quantum value. To accommodate this we give Alice $N$ settings $\{0,1,2,...,N-1\}$, where $A_0=A_Z$ is simultaneously both a chained Bell measurement and an alignment measurement, and give Bob the $N+1$ settings $\{\emptyset\}\cup\{0, 1, 2,...,N-1\}$ comprising a separate alignment measurement $B_\emptyset = B_Z$ appended to the $N$ chained Bell measurements.

Strategies utilizing at-most-bipartite nonlocal resources will be unable to perform well at both the alignment task and the chained Bell task simultaneously. This type of tradeoff for the special case of the CHSH game is observed in the argument of \cite{Chao17} as well as the more recent works of \cite{coiteux:2021,cao:2022,mao:2022}. Our inflation treatment is inspired by these more recent works, but the generalization is not immediate and requires a modified approach, due to the fact that while there is a single set of measurement configurations that can be applied to either one of $\ket{\Phi^+}$ or $\ket{\Phi^-}$ to obtain a maximal CHSH violation (albeit for a different symmetry of the inequality depending on the measured state), this does not cleanly extend to higher-order chained Bell scenarios where there turns out to not be a single set of measurement configurations that lead to chained Bell violations for the two different states while \textit{also} including---crucially---the computational basis setting $Z$ as one of Alice's settings. (It is possible to implement a set of chained Bell measurements such that either $\ket{\Phi^+}$ or $\ket{\Phi^-}$ will result in a chained Bell violation, by measuring along the equator of the Bloch sphere as in for example \cite{tan:2017}---but this omits the computational basis measurement). Our strategy to deal with this complication is to include the settings pairs that appear in both $\mathsf{CB}$ and $\mathsf{CB}'$ within predicates for our game; then, if the $C_X$ outcome is the ``wrong" one for the chained Bell setting pair of Alice and Bob, the win condition defaults to accepting any Alice/Bob output result.

Keeping this general strategy in mind, let us now derive an inequality with the inflation technique. We will use the same inflation as in Figure 2 of \cite{coiteux:2021}
, but include more measurement settings for Alice and Bob to accommodate the chained Bell game; see Figure~\ref{fig:Inflation_NearPTGame} below for a depiction of the inflation. 
 The ``original" variables appear on the top half of Figure~\ref{fig:Inflation_NearPTGame}; let us distinguish the ``cut" inflation variables which appear on the bottom half (where it is the $\mathsf A$-$\mathsf C$ link that is cut) with an asterisk: $A^*$, $B^*$, $C^*$. Now let $\mathsf{CBSQ}^*$ stand for the right side of \eqref{e:cbmono} and $\mathsf{CB}'\mathsf{SQ}^*$ stand for the right side of \eqref{e:cbmonoprime}, with a single change that $C_0$ is replaced with $C_0^*$ while all the $A$/$B$ variables are unstarred. Here $\mathsf{SQ}^*$ indicates a reference to the ``sameness quantity" $2\prbP(A_0=C_0^*)$. 

In the argument below, we will have to be precise in our use of conditional probabilities, and so will write them out explicitly. Where we have written, for instance, $\prbP(A_1=B_0)$, this is really a conditional probability that Alice and Bob have the same outcome, conditioned on Alice choosing setting $1$ and Bob choosing setting $0$. Let us write this now as $\prbP(A=B|T=1,U=0)$, using the convention that the settings associated to $A/B/C$  are $T/U/V$, respectively, and starred setting variables like $V^*$ will be the setting of the corresponding starred parties (so $C^*$, for $V^*$). 

With this notation in hand, we now let $\left.\mathsf{CBSQ}^*\right|_{C_1=+1}$ represent the $\mathsf{CBSQ}^*$ quantity where all probabilities are additionally conditioned on the event that $C$ chooses setting $1$ \textit{and} observes outcome +1; let $\left.\mathsf{CB}'\mathsf{SQ}^*\right|_{C_1=-1}$ be defined analogously. So for example the first term of $\left.\mathsf{CB}'\mathsf{SQ}^*\right|_{C_1=-1}$ is $\prbP(A=B|T=0,U=N-1,V=1,C=-1)$ and the last term is $\prbP(A=C^*|T=0,V^*=0, Z=1,C=-1)$. Note we are using a different output set $\{+1,-1\}$ for the $C_1$ variable; this is to serve as a better visual reminder of how $C_1=C_X$ projects the Alice/Bob state into one of $|\Phi^+\rangle$ or $|\Phi^-\rangle$ depending on the measurement's outcome. In the degenerate case that the probability of the given $C_1$ outcome is zero, one can re-define the corresponding conditional $\mathsf{CBSQ}^*$ quantity as zero, which will maintain consistency in the remainder of the argument. The bounds \eqref{e:cbmono} and \eqref{e:cbmonoprime} will still apply to these quantities, because the distribution conditioned on the setting and outcome of $C_1$ is still a nonsignaling distribution, and only the nonsignaling property was required to derive these monogamy bounds. (That a setting-and-outcome-conditional distribution of the other parties is nonsignaling is justified in \cite{coiteux:2021} by appealing to the space-like separation of the parties; this can also be formally derived from the standard no-signaling principal \cite{bierhorst:2024}.) 

\begin{figure}[!tbh]
\centering
\includegraphics[width=0.48\linewidth]{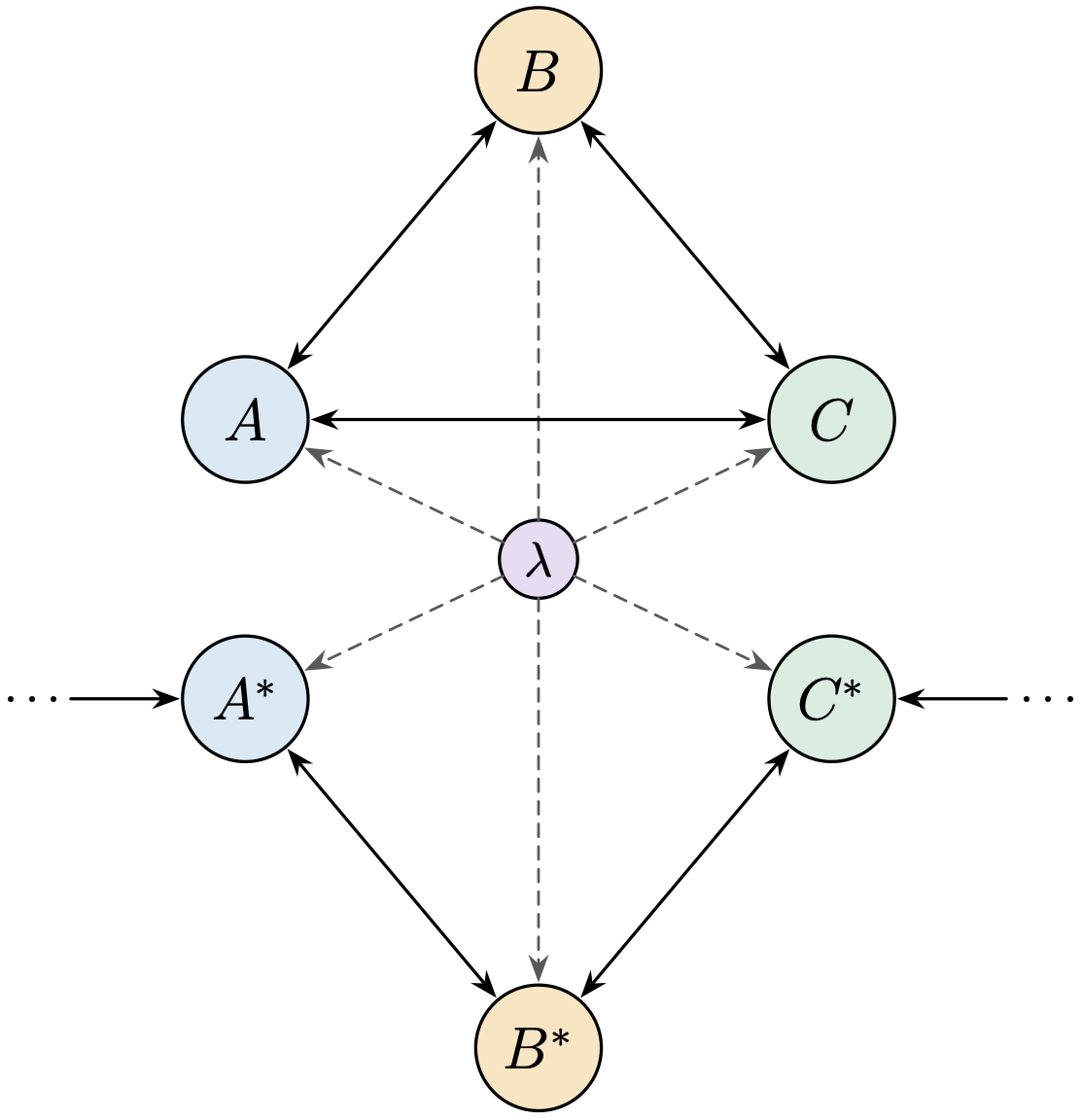}
\caption{Schematic depiction of the (cut) inflation, inspired by Figure 2 of~\cite{coiteux:2021}, applied to the chained-Bell-inspired family of games exhibiting arbitrarily-close-to-pseudotelepathic behavior. The original variables appear in the top half, and the cut inflation variables (denoted \(A^\ast,\,B^{\ast},\,C^{\ast}\)) appear in the bottom half. Solid bidirectional arrows between any two variable-nodes represent bipartite resources shared between the corresponding parties. The dashed arrows incident into each variable-node from the node labelled \(\lambda\) represent globally shared classical randomness.}
\label{fig:Inflation_NearPTGame}
\end{figure}

Now a key manipulation, following the last page of the supplementary material of \cite{cao:2022}, uses the fact that if we have two quantities less than or equal to $2N+1$, then any convex conbination of the two quantities obeys this bound as well. So even if we don't know the value of $\prbP(C=+1|V=1)$, we do know that $\prbP(C=-1|V=1)=1-\prbP(C=+1|V=1)$, and so the quantity $\left.\mathsf{CBSQ}^*\right|_{C_1=+1}\prbP(C=+1|V=1)+\left.\mathsf{CB}'\mathsf{SQ}^*\right|_{C_1=-1}\prbP(C=-1|V=1)$ represents such a convex combination, and so is bounded above by $2N+1$. This has the effect of moving the outcome of $C$ from the conditioner, into the argument of the conditional probability terms in these quantities, noticing that the no-signaling principle allows us to add settings choices for $A$,$B$, and $C^*$ at will to the conditioner of $\prbP(C=-1|V=1)$ to enable the needed manipulation. (We also note that the $2N+1$ bound on the convex combination holds for the degenerate case where $C_1$ takes one value with probability one and the other with probability zero, and so all conditional probabilities are now well-defined with $C_1$ removed from the conditioner.) We can write the bounded quantity as
\begin{multline*}
\prbP(A=B,C=+1|T=0,U=0,V=1)+\prbP(A=B,C=+1|T=1,U=0,V=1)\\
+\prbP(A=B,C=+1|T=1,U=1,V=1)+\prbP(A=B,C=+1|T=2,U=1,V=1)\\
+\cdots +\prbP(A=B, C=+1|T=N-1, U= N-1, V=1)+\prbP(A\ne B, C=+1|T=0,U=N-1,V=1)\\
+\prbP(A\ne B,C=-1|T=0,U=N-1,V=1)+\prbP(A\ne B,C=-1|T=1,U=N-1,V=1)\\
+\prbP(A\ne B,C=-1|T=1,U=N-2,V=1)+\prbP(A\ne B,C=-1|T=2,U=N-2,V=1)\\
+\cdots +\prbP(A\ne B, C=-1|T=N-1,U=0,V=1)+\prbP(A= B, C=-1|T=0,U=0,V=1)\\
+2\prbP(A=C^*, C=+1|T=0,V^*=0,V=1)+2\prbP(A=C^*, C=-1|T=0,V^*=0,V=1),
\end{multline*}
where the top three lines correspond to $\mathsf{CB}$ terms, the next three lines $\mathsf{CB}'$ terms, and the last line collects the ``sameness" quantities from the monogamy expressions. Note there are precisely two settings configurations that occur in the first three lines while also re-occurring in the next three lines; these are $T=0,U=0,V=1$ and $T=0,U=N-1,V=1$ and so each pair of matching terms can be summed together into $\prbP(A=B|T=0,U=0,V=1)=\prbP(A=B|T=0,U=0)$ and $\prbP(A\ne B|T=0,U=N-1,V=1)=\prbP(A\ne B|T=0,U=N-1)$. Since the two ``sameness" quantities can similarly be combined together, we can write the expression more compactly as 
\begin{equation*}
\begin{split}
&\prbP(A=B|X=0,Y=0)+\prbP(A\ne B|T=0,U=N-1)
+\sum_{i=1}^{N-1}\sum_{j\in{0,1}}\prbP(A=B,C=+1|T=i,U=i-j,V=1) \\
& + \sum_{i=1}^{N-1}\sum_{j\in{0,1}}\prbP(A\ne B,C=-1|T=i,U=N-i-j,V=1)
+\prbP(A=C^*|T=0,V^*=0).
\end{split}
\end{equation*}
The above quantity is bounded above by $2N+1$. So if we add $2(N-1)$ copies of $\prbP(C=+1|V=1)+\prbP(C=-1|V=1)$, each of which equals one, the following is bounded above by $4N-1$:
\begin{equation*}
\begin{split}
&\prbP(A=B|T=0,U=0)+\prbP(A\ne B|T=0,U=N-1) \\
& + \sum_{i=1}^{N-1}\sum_{j\in{0,1}}\prbP(\{A=B \text{ and }C=+1\}\text{ or }\{C=-1\}|T=i,U=i-j,V=1) \\
& + \sum_{i=1}^{N-1}\sum_{j\in{0,1}}\prbP(\{A\ne B \text{ and }C=-1\}\text{ or }\{C=+1\}|T=i,U=N-i-j,V=1) + 2\prbP(A=C^*|T=0,V^*=0).
\end{split}
\end{equation*}
We now finish the argument with an inflation technique approach in the spirit of~\cite{coiteux:2021,cao:2022,mao:2022}, following the arrangement of parties depicted in Figure~\ref{fig:Inflation_NearPTGame}, using the $\sim$ symbol to indicate the equality of marginal distributions for subsets of parties that share isomorphic resource ancestry:
\begin{align}\label{e:cbinf}
    \prbP(A=C^*|T=0,V^*=0)&\sim
    \prbP(A^*=C^*|T^*=0,V^*=0)\nonumber\\
    &=    \prbP(A^*=C^*|T^*=0,U^*=\emptyset, V^*=0)\nonumber\\
    &\ge
\prbP(A^*=B^*|T^*=0,U^*=\emptyset, V^*=0)+\prbP(B^*=C^*|T^*=0,U^*=\emptyset, V^*=0)-1\nonumber\\
&\sim
\prbP(A=B|T=0,U=\emptyset, Z=0)+\prbP(B=C|T=0,U=\emptyset, V=0)-1\nonumber\\
&=
\prbP(A=B|T=0,U=\emptyset)+\prbP(B=C|U=\emptyset, V=0)-1.
\end{align}
Above, the inequality is an application of \eqref{e:pqrprob1} after relabeling $R$'s outcomes to convert $R\ne T$  and $R\ne S$ to equality conditions, and the equalities on the second and fifth line follow from the no-signaling principle. 
All together, we obtain our final result, which is that the following is bounded above by $4N+1$:
\begin{equation}\label{e:cbquant}
\begin{split}
&\prbP(A=B|T=0,U=0)+\prbP(A\ne B|T=0,U=N-1)\\
+&\sum_{i=1}^{N-1}\sum_{j\in{0,1}}\prbP(\{A=B \textnormal{ and }C=+1\}\textnormal{ or }\{C=-1\}|T=i,U=i-j,V=1)\\
+&\sum_{i=1}^{N-1}\sum_{j\in{0,1}}\prbP(\{A\ne B \textnormal{ and }C=-1\}\textnormal{ or }\{C=+1\}|T=i,U=N-i-j,V=1)\\
+&\,2\prbP(A=B|T=0,U=\emptyset)+2\prbP(B=C|U=\emptyset, V=0).
\end{split}
\end{equation}
Expression \eqref{e:cbgame} in the main text is obtained by extracting the predicates from the above probabilities, relabeling the $C_X$ outcomes to follow the binary $\{0,1\}$ convention used throughout the rest of the paper. We show in the next section that when these predicates are queried equiprobably, appropriate quantum measurements of a GHZ state can lead to an expected winning probability of
\begin{equation}\label{e:cbquantwinprob}
\frac{3N + 2 + N\cos\left(\pi/2N\right)} {4N + 2} \ge \frac{4N + 2 - \frac{\pi^2}{8N}} {4N + 2},
\end{equation}
where the expression on the right follows from the inequality $\cos(x)\ge 1-x^2/2$. This is asymptotically $\epsilon$-close to pseudotelepathic compared to the LOSR-GMNL bound of $(4N+1)/(4N+2)$ that must be obeyed by networks of underlying bipartite resources, in a sense that will be made precise in Section \ref{s:limitingbehavior}.

\subsection{\label{ss:QuantumVal_ChainedBell}Calculation of Quantum Value for Chained Bell Quantities \texorpdfstring{\(\mathsf{CB}\)}{CB} and \texorpdfstring{\(\mathsf{CB}'\)}{CB1}}


We start by reviewing a well known expression (see for instance \cite{CHAINED}) for the correlator values that result from Alice and Bob measuring Bell states $\ket{\Phi^{\pm}} = (\ket{00}\pm\ket{11})/2$ with Pauli observables in the $X$/$Z$ plane:
\begin{equation}\label{e:measure}
\langle A_\alpha, B_\beta\rangle =\bra{\Phi^{\pm}} M_\alpha \otimes M_\beta \ket{\Phi^{\pm}} = \cos(\alpha\mp \beta),\text{ where } M_\theta = \cos\theta Z+\sin \theta X=\begin{bmatrix}\cos \theta & \sin \theta \\ \sin \theta & -\cos \theta\end{bmatrix}
\end{equation}
Above, $\theta$ is the angle on the Bloch sphere extending down from the $Z$ axis, towards the $X$ axis. The identity \eqref{e:measure} can be obtained directly by writing out and computing the matrix algebra and applying sum/difference of cosine formulas. Since $\langle A_\alpha, B_\beta\rangle = \prbP(A_\alpha = B_\beta) - \prbP(A_\alpha \ne B_\beta)$, \eqref{e:measure} leads to the following expressions:
\begin{equation}\label{e:cosprobs}
\prbP_{\ket {\Phi^\pm}}(A_\alpha = B_\beta)= \frac{1+\cos(\alpha\mp \beta)}{2},\quad\quad
\prbP_{\ket {\Phi^\pm}}(A_\alpha \ne B_\beta)= \frac{1-\cos(\alpha\mp \beta)}{2}
\end{equation}
Accordingly, we choose the following $N$ measurement bases for Alice and Bob, which are depicted in Figure~\ref{fig:BlochSphereAngles}: 
\begin{equation}\label{e:cbangles}
A_i = M_\alpha  \text{ with } \alpha=i\left(\frac{\pi}{N}\right),\quad\quad
B_j = M_\beta \text{ with } \beta=j\left(\frac{\pi}{N}\right)+\frac{\pi}{2N}
\end{equation}
Then in computing the chained Bell quantity $\mathsf{CB}$, when the state is $\ket{\Phi^+}$, we find that for $i=0,...,N-2$ we have
\begin{equation*}
\prbP(A_i=B_i)=\prbP(A_{i+1}=B_i)=\frac{1+\cos(\pi/2N)}{2}=\prbP(A_{N-1}=B_{N-1})=\prbP(A_0\ne B_{N-1})
\end{equation*}
by direct computation from \eqref{e:cosprobs}, enlisting the elementary identities $\cos(x)=\cos(-x)$ and $\cos(\pi-x)=-\cos(x)$. This leads to a $\mathsf{CB}$ value of $2N$ times the above quantity, or 
\begin{equation}\label{e:cbvalue}
\mathsf{CB}=N+N\cos\left(\frac{\pi}{2N}\right)\ge 2N-\frac{\pi^2}{8N}
\end{equation}
using the inequality $\cos(x)\ge 1-x^2/2$.

Now, if the state is $\ket{\Phi^-}$, we can use the same Alice and Bob measurements as given in \eqref{e:cbangles}, but we need to examine the different pairings of these measurements that appear in $\mathsf{CB}'$. The sum of probabilities also comes out to equal the value displayed on the right side of \eqref{e:cbvalue}; this follows after observing that all terms of $\mathsf{CB}'$ in \eqref{e:chainbellprime}, other than the last term $\prbP(A_0=B_0)$, are of the form $\prbP(A_x\ne B_{N-x})$ or $\prbP(A_x\ne B_{N-x-1})$; referring to \eqref{e:cosprobs} we see that $\cos(\alpha+\beta)$ in the expression for $\prbP_{\ket{\phi^-}}(A_\alpha \ne B_\beta)$ will always be equal to  $\cos(\pi\pm(\pi/2N))=-\cos(\pi/2N)$ when we use the angles in \eqref{e:cbangles}, leading $\prbP_{\ket{\phi^-}}(A_\alpha \ne B_\beta)$ to equal $(1+\cos(\pi/2N))/2$ as desired. The last term $\prbP_{\ket{\phi^-}}(A_0=B_0)$ is also immediately seen to equal this quantity upon substituting the angles in \eqref{e:cbangles}.

\begin{figure}[!tbh]
\centering
\includegraphics[width=0.48\linewidth]{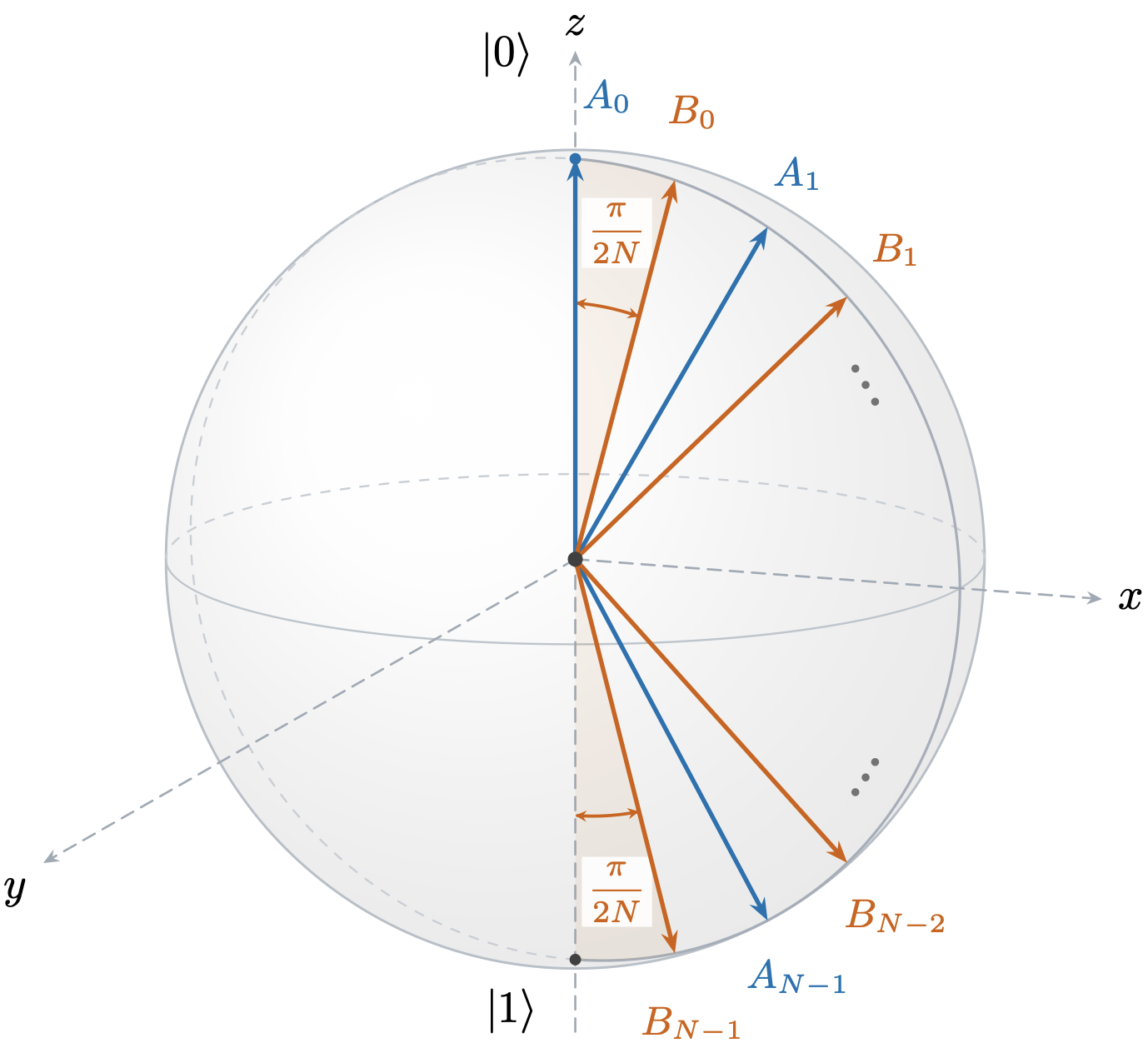}
\caption{\label{fig:BlochSphereAngles}Bloch-sphere geometry of the interleaved qubit observables \(A_i=M_{\alpha_i}\) and \(B_j=M_{\beta_j}\), where \(M_\theta=\cos\theta\,\sigma_Z+\sin\theta\,\sigma_X\),  \(\alpha_i=i\pi/N\), and \(\beta_j=(j+\tfrac12)\pi/N\), for \(i,j=0,\ldots,N-1\), as in~\eqref{e:cbangles}. The displayed vectors are the \(+1\) eigenstate Bloch directions. Alice's settings (blue) and Bob's settings (orange) alternate in the \(x\)-\(z\) plane with angular spacing \(\pi/(2N)\). \(A_0=B_{\emptyset}\) points along `\(+z\)', corresponding to \(|0\rangle\), while \(B_{N-1}\) remains an angle \(\pi/(2N)\) short of the `\(-z\)' direction \(|1\rangle\). Intermediate settings are omitted for clarity.}
\end{figure}

Now we can calculate the quantum value of \eqref{e:cbquant}. The last line of this expression, written compactly as $\prbP(A_0=B_\emptyset)+2\prbP(B_\emptyset=C_Z)$, is equal to four since all parties are measuring the $GHZ$ state in the $Z$ basis. For the first line $\prbP(A_0=B_0)+\prbP(A_0\ne B_{N-1})$, this could be calculated directly from the GHZ state, but also can be calculated by considering that if $C$ measures $X$, the Alice-Bob state is steered to one of the two Bell states $\ket {\phi^\pm}$ which in either case yields $(1+\cos(\pi/2N))/2$ for each of these probabilities as discussed above, and so these two probabilities must also each equal this value when not conditioned on $C_X$. Finally, the remaining chained Bell probabilities can also be calculated by the law of total probability as for example (for any $1\le i\le N-1$ and $j\in \{0,1\}$)
\begin{align*}
&\prbP(\{ C_X=+1; A_i=B_{i-j}\}\cup\{C_X=-1; \textnormal{ any }A_iB_{i-j}  \})\\
=&\,\prbP(A_i=B_{i-j}|C_X=+1)\prbP(C_X=+1)+P(\textnormal{any value of }A_i,B_{i-j}|C_X=-1)\prbP(C_X=-1)  \\
=&\frac{1+\cos(\pi/2N)}{2}\cdot\frac{1}{2}+1\cdot\frac{1}{2} \\
=&\frac{3+\cos(\pi/2N)}{4}.
\end{align*}
There are $4(N-1)$ such terms as above, leading to an overall quantum value of \eqref{e:cbquant} of $3N+2+N\cos(\pi/2N)$, which yields the expression in \eqref{e:cbquantwinprob}.

\subsection{Limiting Behavior of the Three-Party GHZ-Chained-Bell Game as \texorpdfstring{\(N,t\to\infty\)}{NtInfty}}\label{s:limitingbehavior}

The expression \eqref{e:cbquantwinprob} gives the winning probability of a single trial of the GHZ-chain-Bell game. We simply exponentiate this quantity by $t$ to obtain the expression for the probability that in an experiment of $t$ trials, the honest (quantum) implementation observes $t$ successes, and thus exhibits pseudotelepathic behavior. This probability is (at least)
\begin{equation}\label{e:completeness}
\left(\frac{4N + 2 - \frac{\pi^2}{8N}} {4N + 2}\right)^t
=\left(1-\frac{\pi^2 }{32N^2 + 16N}\right)^t>\left(1-\frac{\pi^2 }{32N^2}\right)^t.
\end{equation}
Given a completeness parameter $\epsilon_c>0$, one can meet the request that the probability of protocol success exceeds $1-\epsilon_c$ by choosing appropriate values of $t$ and $N$ in \eqref{e:completeness} -- informally, $t$ must grow more slowly than $N^2$ to ensure the value of \eqref{e:completeness} tends to 1. Turning to the probability of protocol success under the \textit{foil} theory, which we want to keep small (less than $\epsilon_s$) -- this is effectively the p-value for rejecting the null hypothesis of an underlying at-most-bipartite network, given $t$ successes in $t$ trials -- this is given by 
\begin{equation}\label{e:soundness}
\left(\frac{4N + 1} {4N + 2}\right)^t=\left(1-\frac{1} {4N + 2}\right)^t\le \left(1-\frac{1} {6N }\right)^t,
\end{equation}
which we can make arbitrarily small by ensuring $t$ grows more quickly than $N$. If we set $t=N^{3/2}$, we can write
\begin{equation*}
\lim_{N\to \infty}\left(1-\frac{\pi^2 }{32N^2}\right)^t=1 \qquad\quad\textnormal{ and }\qquad\quad\lim_{N\to \infty}\left(1-\frac{1} {6N }\right)^t=0.
\end{equation*}
This means that given two values of $\epsilon_c,\epsilon_s>0$, there exist choices of $N$ and $t$ such that for the $N$th member of our family of experiments, an experiment of $t$ trials will observe $t$ successes with probability exceeding $1-\epsilon_c$ according to the noiseless quantum prediction, and supply a p-value below $\epsilon_s$ for rejecting the null hypothesis of an underlying bipartite-only network. 

\subsection{Generalization To \texorpdfstring{\(n|(n-1)\)}{nnminus1} Separation} \label{s:nsep}

By distributing the role of Charlie to multiple Charlie-like parties and measuring a higher order GHZ state, one can adapt the construction of the above sections to a family of $n$-party quantum behaviors with $\epsilon$-close-to-perfect winning probabilities, for which the behavior cannot be replicated by an underlying network of $(n-1)$-partite resources according to the LOSR-GMNL definition. A similar adaptation strategy for the CHSH case can be found in the earlier work of \cite{Chao17} and the more recent inflation technique arguments of \cite{coiteux:2021a,cao:2022,mao:2022}. We outline how this strategy can be applied to generalize our GHZ-chained-Bell game as follows:
retain the same measurements for $\mathsf A$ and $\mathsf B$, and replace $\mathsf C$ with a collection of $n-2$ parties $\mathsf C^1,...,\mathsf C^{n-2}$, each measuring $Z$ and $X$. Let the measured state be the $n$-party GHZ state $(\ket 0^{\otimes n}+\ket 1^{\otimes n})/\sqrt 2 $. When all ``Charlies" measure $X$, the $\mathsf A$-$\mathsf B$ state is once again steered to one of $\ket {\Phi^+}$ or $\ket {\Phi^-}$, with the outcome now depending on whether an even or odd number of $C^i_X$ outcomes are $-1$. Therefore in the argument of Section \ref{s:cbinflationproof}, we accordingly replace the conditional $\mathsf{CB}'\mathsf{SQ}^*$ quantities with expressions of the form
$\left.\mathsf{CB}'\mathsf{SQ}^*\right|_{\langle \vec{C}_X\rangle =+1}$ and $\left.\mathsf{CB}'\mathsf{SQ}^*\right|_{\langle \vec{C}_X\rangle =-1}$, with $\langle \vec{C}_X\rangle =\pm 1$ representing the parity condition for the collective of Charlie parties, and the role of the starred Charlie party appearing in the ``sameness quantity" terms is assumed by the final Charlie party $\mathsf {C}^{n-2}$. We replace the 3-party cut inflation used in Section \ref{s:cbinflationproof} with an $n$-party version that can be seen in Figure 2 of the longer companion paper of Coiteux-Roy, Wolfe, and Renou \cite{coiteux:2021a}, for which we order the parties as $\mathsf{A}, \mathsf{B}, \mathsf {C}^1,...,\mathsf{C}^{n-2}$ such that adjacent pairs of parties share the full resource heritage spanning all $(n-1)$-partite resources relevant to both of them, but $\mathsf{A}$ and $\mathsf{C}^{n-2}$ are causally disconnected. The argument continues as before with all instances of $C_1=C_X$ replaced with $\langle \vec{C}_X\rangle$ conditions, up to the inflation argument of \eqref{e:cbinf}; now the single application of the bound \eqref{e:pqrprob1} is replaced by $n-2$ iterated applications (note here the parallel to expression (21) of \cite{coiteux:2021a}), to relate the probability of $A^*=C^{n-2,\ast}$ to the sum of the probabilities of $A^*=B^*, B^*=C^{1,\ast}, C^{1,\ast}=C^{2,\ast}, \ldots, C^{n-3,\ast}=C^{n-2,\ast}$, all of which are isomorphic to the un-starred versions of these pairwise conditions. The final result is an adapted version of the predicates in \eqref{e:cbgame} that substitutes instances of $C_X$ with $\langle \vec{C}_X\rangle$, and replaces the final ``sameness" term $B_\emptyset =C_Z$ with $n-2$ corresponding equality terms $B_\emptyset =C^1_Z, C^1_Z =C^2_Z,\ldots, C^{n-3}_Z =C^{n-2}_Z$, spanning the $n-2$ different Charlie parties.

\end{document}